\documentclass[11pt,a4paper]{article}
\usepackage{jheppub}
\usepackage[utf8]{inputenc}
\usepackage[normalem]{ulem}
\usepackage{epsfig}
\usepackage{graphicx}
\usepackage{amsfonts,amsmath,amsthm,amssymb}
\usepackage{epiolmec}
\usepackage{phaistos}
\usepackage{adjustbox}
\newtheorem{theorem}{Theorem}

\newtheorem{corollary}{Corollary}
\newtheorem{remark}{Remark}
\newtheorem{proposition}[]{Proposition}
\newtheorem{lemma}{Lemma}

\usepackage{mathtools}
\numberwithin{equation}{section}
\newcommand{\tr}{\operatorname{tr}}
\usepackage{ulem}
\usepackage{multicol}
\usepackage{latexsym}
\usepackage{mathrsfs} 
\newcommand\Vtextvisiblespace[1][.5em]{%
  \mbox{\kern.06em\vrule height.3ex}%
  \vbox{\hrule width#1}%
  \hbox{\vrule height.3ex}}
\usepackage{cancel}

\usepackage{xparse}
\usepackage{float}
\usepackage{todonotes}
\usepackage{etoolbox}
\usepackage{tikz}
\usepackage{bbold}
\usepackage{upgreek}
\usetikzlibrary{arrows,positioning,cd,calc,math,decorations.pathreplacing,decorations.markings,decorations.pathmorphing,patterns
}
\tikzset{wave/.style={decorate, decoration=snake}}

\usepackage{xcolor}
\definecolor{MyBlue}{rgb}{0.25,0.5,0.75}
\colorlet{NextBlue}{MyBlue!20}
\colorlet{SecondBlue}{MyBlue!40}
\usepackage{subcaption}
\NewDocumentCommand{\tens}{t_}
 {%
  \IfBooleanTF{#1}
   {\tensop}
   {\otimes}%
 }
\NewDocumentCommand{\tensop}{m}
 {%
  \mathbin{\mathop{\otimes}\displaylimits_{#1}}%
 }

\DeclareMathOperator{\res}{Res}
\newcommand{\Li}{\mathrm{Li}}

\newcommand{\Yt}{\widetilde{Y}}
\newcommand{\sigmat}{\widetilde{\sigma}}

\newcommand{\Qt}{\widetilde{Q}}
\newcommand{\Pt}{\widetilde{P}}

\newcommand{\Ht}{\widetilde{H}}

\newcommand{\Tt}{\widetilde{\mathcal{T}}}
\newcommand{\Zt}{\widetilde{Z}}
\newcommand{\ttau}{\widetilde{\tau}}
\newcommand{\ta}{\widetilde{a}}
\newcommand{\teta}{\widetilde{\eta}}
\newcommand{\tnu}{\widetilde{\nu}}
\newcommand{\Gt}{\widetilde{G}}
\newcommand{\rhot}{\widetilde{\rho}}
\newcommand{\cBt}{\widetilde{\mathcal{B}}}

\newcommand{\T}{\mathcal{T}}

\newcommand{\e}{\varepsilon}
\newcommand{\dd}{{\rm d}}
\newcommand{\taut}{\widetilde{\tau}}
\newcommand{\at}{\ta}
\newcommand{\nut}{\tnu}
\newcommand{\cA}{\mathcal{A}}
\newcommand{\cM}{\mathcal{M}}
\newcommand{\cB}{\mathcal{B}}

\usepackage{tikz}
\usepackage{pgfplots}
\usepackage{upgreek}
\usetikzlibrary{decorations.markings,pgfplots.fillbetween}
\definecolor{darkspringgreen}{rgb}{0.05, 0.5, 0.06}
\definecolor{MyBlue}{rgb}{0.25,0.25,0.75}
\definecolor{MyRed}{rgb}{0.75,0.25,0.25}
\colorlet{NextBlue}{MyBlue!20}
\colorlet{SecondBlue}{MyBlue!40}

\newcommand{\sigmab}{\widetilde{\sigma}}
\definecolor{lightbeige}{RGB}{245, 245, 220} 

\title{Modular transformations of tau functions and conformal blocks on the torus}
\abstract{The connection problem for isomonodromic tau functions on the one-punctured torus concerns the ratio between the tau function and its modular transform, associated to dual pants decompositions of the torus. In this paper, we study the modular transformations of the tau function and consequently derive the connection constant. Moreover, through the relation with two-dimensional Conformal Field Theory, we also obtain an exact closed formula for the $c=1$ Virasoro modular kernel, whose expression was previously unknown, and relate it to the $c\rightarrow\infty$ (semiclassical) modular kernel and $SL_2(\mathbb{C})$ complex Chern-Simons amplitudes. Finally, we prove that the connection constant and the two, $c=1$ and $c\to \infty$, modular kernels are generating functions of canonical transformations on the character variety of the one-punctured torus. Our results are also relevant for the $\mathcal{N}=2^*$ gauge theory.}

\author{Fabrizio Del Monte\(^a\), Harini Desiraju\(^{b, c}\), Pavlo Gavrylenko\(^{d, e}\)}

\affiliation{\({}^a\)School of Mathematics, University of Birmingham, Watson Building, Edgbaston, Birmingham B15 2TT, UK}
\affiliation{\({}^b\)Sydney Mathematical Research Institute, A14 Quadrangle, Camperdown campus, University of Sydney, Sydney 2010, Australia}
\affiliation{\({}^c\)Mathematical Institute, University of Oxford, Oxford OX2 6GG, UK}
\affiliation{\({}^d\)SISSA, Via Bonomea 265, 34136 Trieste, Italy}
\affiliation{\({}^e\)Bogolyubov Institute for Theoretical Physics, 14-B Metrolohichna str., Kyiv 03143, Ukraine }
    
    \emailAdd{f.delmonte.mp@gmail.com}
    \emailAdd{harini.desiraju@gmail.com}
    \emailAdd{pasha.145@gmail.com}

\date{}

\begin{document}

\maketitle
\section{Introduction}

Isomonodromic deformation equations (IDEs), of which Painlev\'e equations are special cases, play a fundamental role in mathematical physics, with applications ranging from algebraic geometry to statistical mechanics. The central object of study in these equations is the tau function, which is deeply connected to the solutions of IDEs and acts as the generator of their Hamiltonians. A classical challenge in the theory of integrable systems is the connection problem, which seeks to understand the asymptotics of the tau function at one critical point based on its behaviour at another. This problem was first sought after due its ties to statistical physics. For instance, the connection constant arising in the Ising model—equivalent to that of a special solution of Painlev\'e III—was determined in \cite{tracy1991asymptotics}. Since then, there have been several approaches to solve this problem, but for generic IDEs the solution has long been elusive. 

A key aspect that makes this problem hard to solve is the lack of a closed-form expression for the tau function. In the past decade, the  relation between Painlev\'e equations and two-dimensional Conformal Field Theories (CFTs) started to become clear through the works of \cite{GavrylenkoLisovyi2016b} and others. Due to this correspondence, the Painlev\'e tau functions (in most of the cases) are expressible in terms of conformal blocks, which are the basis of representations for the Virasoro algebra that encodes the symmetries of a given CFT. Such an expression led to the resolution of the connection problem for the cases of Painlev\'e II, III, and VI \cite{Iorgov2013,Its2014,ILP2016}. 

In \cite{DMDG2020} (see also \cite{Bonelli2019,bonelli2021a} for previous physical constructions), we extended the Kyiv formula to the case of Fuchsian systems on surfaces of genus 1. Specifically, we showed that the associated tau function is expressible as a Fredholm determinant of a trace class operator explicitly described in terms of hypergeometric functions. A key feature of our construction, which can be viewed as the genus 1 generalization of \cite{GavrylenkoLisovyi2016b}, was the pants decomposition of the torus along its $A$-cycle. This effectively let us describe the local behavior of Fuchsian systems on the torus in terms of Fuchsian systems on three-punctured spheres, which are explicitly solvable in terms of hypergeometric functions. The resultant tau function turns out to be suitable to study the asymptotics for $\tau\to i \infty$. In the case of one puncture that we will consider in this paper, this tau function depends on the modular parameter $\tau$ and the monodromy variables.
This tau function is well defined up to a constant independent of $\tau$. 

As a next step in the analytical understanding of these tau functions, in \cite{DMDG2022} we used our Fredholm determinant representation to study their monodromy dependence
by calculating the {\it total} logarithmic derivatives, with respect to both times and monodromy variables. In this way, the normalisation is encoded into a closed one-form on the monodromy variety, an idea originating from \cite{ILP2016} for the case of a sphere with singular points. By identifying this one-form with a symplectic potential for the symplectic form on the space of monodromy data, we highlighted the role of the tau function as the generator of the monodromy symplectomorphism, generalising the results of \cite{BertolaKorotkin2019,ItsProkhorov2018} to the case of genus 1 surfaces.  
In this paper, we tackle the final open problem left by the above construction, namely the explicit integration of the difference between the two one-forms normalising tau functions defined by different pants decomposition, thus obtaining the explicit form of the connection constant between the $\tau\rightarrow i\infty$ and $\tau\rightarrow 0$ asymptotics of the isomonodromic problem. As we show in this paper, the connection constant is closely related to the modular transformation $\tau \to -1/\tau$ of the tau function. Through the aforementioned relation between the tau function and $c=1$ Virasoro conformal blocks, we also obtain the first explicit formula for the $c=1$ Virasoro modular kernel, describing modular transformations of these conformal blocks. As a consequence of the AGT correspondence \cite{Alday:2009aq} and our previous work \cite{Bonelli2019,desiraju2022painleve}, the connection constant and the modular kernel can also be seen to describe the modular transformation of the dual partition function and the Nekrasov instanton partition function respectively \cite{Nekrasov2006}, for the $\mathcal{N}=2^*$ gauge theory. Finally, we give a new explicit relation between the $c=1$ and $c\to \infty$ Virasoro modular kernel, which has a natural interpretation in terms of the symplectic geometry of the underlying character variety.

\textbf{Outline of the paper. } The paper is organised as follows. In Section \ref{sec:setup}, we give a brief overview of isomonodromic deformations on a once-punctured torus, and of relevant previous results on the associated tau function. Section \ref{sec:pants} is devoted to the computation of the asymptotics of the isomonodromic problem in the two dual regimes, $\tau\rightarrow i\infty$ and $\tau\rightarrow 0$, associated to two dual pants decompositions. This allows us to compute the modular connection constant of the tau function in Section \ref{sec:Upsilon}. In Section \ref{sec:SDual}, we interpret our results in terms of the modular properties of the linear system and the tau function. The last two sections use the computation of the modular connection constant to derive new results in two-dimensional Conformal Field Theory: in Section \ref{sec:c1}, we compute the modular kernel for $c=1$ Virasoro conformal blocks, and in Section \ref{sec:cinfty} we relate $c=1$ and $c=\infty$ modular kernels through their relation to the connection constant, while giving all these quantities an interpretation as generating functions on the character variety.

\textbf{Acknowledgements}

This research was partly conducted during the authors' visit to the Okinawa Institute of Science and Technology (OIST) through the Theoretical Sciences Visiting Program (TSVP), for the Thematic Program ``Exact Asymptotics: From Fluid Dynamics to Quantum Geometry''.

FDM and PG thanks the Galileo Galilei Institute for Theoretical Physics for the hospitality, during the program ``BPS Dynamics and Quantum Mathematics'', and the INFN for partial support during the completion of this work.

Part of this project was carried out while H.D. was visiting the Hausdorff
Research Institute for Mathematics (HIM) in 2025,
supported by the Deutsche Forschungsgemeinschaft (DFG, German Research
Foundation) under Germany’s Excellence Strategy EXC-2047/1-390685813. She also acknowledges the generous support of the SMRI Postdoctoral Fellowship, and Marie Skłodowska-Curie Postdoctoral Fellowship \#101203697.

The work of PG was partly supported by the INFN Iniziativa Specifica GAST, and by the MIUR PRIN Grant 2020KR4KN2 ``String Theory as a bridge between Gauge Theories and Quantum Gravity''.  
PG also acknowledges funding from the EU project Caligola (HORIZON-MSCA-2021-SE-01), Project ID: 101086123, and CA21109 - COST Action CaLISTA.

\section{General setup and overview of the results}\label{sec:setup}
In this section we review the relevant setup regarding isomonodromic deformations on the torus (see \cite{DMDG2020,DMDG2022} for more details), and then state our main results. We consider the following isomonodromic system on a once-punctured torus with the singularity at $z=0$:
\begin{align}
  Y(z,\tau)^{-1} \partial_z Y(z,\tau) &=:  L_{z} (z,\tau) =\left(\begin{array}{cc}
        P(\tau) & m\,x(-2Q(\tau),z) \\
        m\,x(2Q(\tau),z) & -P(\tau)
    \end{array}\right),\\
    Y(z,\tau)^{-1} \partial_{\tau} Y(z,\tau) &=: L_{\tau}(z,\tau) =-m\left( \begin{array}{cc}
        0 & y(-2Q,z), \\
        y(2Q,z) & 0
    \end{array} \right), \label{linear_systemCM}
\end{align}
where $x, y$ are the Lam\'e functions defined as 
\begin{align}
    x(\xi, z) = \frac{\theta_{1}(z- \xi) \theta_{1}'(0)}{\theta_{1}(z) \theta_{1}(\xi)}, && y (\xi, z) = \partial_{\xi} x(\xi, z),
\end{align}
and $\theta_1$ is the Jacobi theta function 
\begin{equation}
\theta_1(z):=-i\sum_{n\in\mathbb{Z}}(-1)^ne^{i\pi\tau\left(n+\frac{1}{2} \right)^2}e^{2\pi i\left(n+\frac{1}{2} \right)z}. \label{eq:Theta1def}
\end{equation}
The compatibility condition of the above system yields the Non-Autonomous Elliptic Calogero-Moser (NAECM) model
\begin{align}\label{eq:NAECM}
    2\pi i\frac{\dd P}{\dd\tau}=m^2\wp'(2Q|\tau), && 2\pi i\frac{\dd Q}{\dd\tau}=P.
\end{align}
These equations also arise as the Hamilton equations for the Hamiltonian
\begin{gather}\label{def:Ham}
    H :=\oint_{A}\frac{dz}{2}\tr \left(L_{z}(z,\tau)\right)^2= P^2-m^2\wp(2Q|\tau)-2m^2\eta_1(\tau),
\end{gather}
where 
\begin{align}\label{eta1_def}
    \eta_1(\tau)=-\frac{1}{6}\frac{\theta_1'''(0)}{\theta_1'(0)}.
\end{align}
The generating function for $H$ is called the tau function, and it will be the main object of our study:
\begin{align}\label{def:Tau}
     2\pi i \partial_{\tau} \log \T := H.
\end{align}

The tau function can be described in terms of the functions $P(\tau),\, Q(\tau),\,\tau,\,m$, or equivalently as a function of the monodromy data of the linear system \eqref{linear_systemCM} which we present below. These are part of local coordinates on spaces $\mathcal{A}_{1,1}$ and $\mathcal{M}_{1,1}$ respectively that we define below, and the map between these spaces is known as the (extended) Riemann-Hilbert map. Let us define these coordinates in detail. 

The residue of the Lax matrix in \eqref{linear_systemCM} is given by
\begin{align}\label{def:G}
	\res_{z=0}L_{z}=-m\sigma_1=:G^{-1}m\sigma_3G,\end{align}
where
\begin{align}\label{eq:GG0g}
	G=e^{\frac{1}{2}g\sigma_3}G_0, && G_0:=\left( \begin{array}{cc}
	1 & -1 \\
	1 & 1		
	\end{array} \right).
\end{align}
 The space of Lax matrices $L_z$ in \eqref{linear_systemCM} is defined as 
\begin{equation}\label{eq:A1n}
    \mathcal{A}_{1,1}:=\bigg\{\tau,\,(G, m),\,(P, Q):\tau\in\mathbb{H},\,G\in SL(2),
    \,P,Q\in\mathbb{C}\bigg\}/\sim,
\end{equation}
where $\sim$ is the equivalence relation  $G\rightarrow D G$, where $\,D\in SL(2)$ is diagonal.
The space $\cA_{1,1}$ is parameterized by $m,P,Q,g,\tau$, and $ \dim\mathcal{A}_{1,1}= 5$.
Similarly, to obtain the monodromy variety, consider the solution of the linear system \eqref{linear_systemCM} $Y(z,\tau)$ which has the following monodromies around the $A,B$-cycles, and the puncture $z=0$ respectively:
\begin{align}
    Y(z+1, \tau) = M_A Y(z,\tau), && Y(z+\tau, \tau) = M_B Y(z,\tau) e^{2\pi i Q(\tau) \sigma_3}, && Y(e^{2\pi i } z, \tau) = M_0 Y(z,\tau).
\end{align}
The monodromy matrices above are 
\begin{align}\label{eq:MAM0MB}
	M_A=e^{2\pi ia\sigma_3}, && M_0=C_0 e^{2\pi i m \sigma_3}C_0^{-1}, && M_B= \frac{1}{\sin 2\pi a}\left(
\begin{array}{cc}
 e^{-\frac{i \nu}{2}  }  \sin (\pi  (2 a-m)) & e^{\frac{i \nu }{2}}  \sin (\pi  m) \\
 -e^{-\frac{i \nu}{2}}  \sin (\pi  m) & e^{\frac{i \nu }{2}}  \sin (\pi  (2 a+m)) \\
\end{array}\right),
\end{align}
with the constraint
\begin{gather}
    M_{B}^{-1} M_{A}^{-1} M_{B} M_{A} M_{0} = \mathbb{1}.  \label{eq:n_moncons}
\end{gather}
The space of matrices $M_A,\,M_B,\,M_0$ up to conjugation is known as the character variety. This space carries a Poisson bracket called the Goldman bracket, for which $a,\,\nu$ above are canonical coordinates and $m$ is a Casimir (see Appendix \ref{sec:charvar} for more details). While Poisson, this space is odd-dimensional and so it cannot be symplectic. We will then work with an extension of this space, known as the extended character variety 
(see \cite{BertolaKorotkin2019} for the genus zero case)
 \begin{align}\label{eq:M1n}
    \mathcal{M}_{1,1} = \left\lbrace M_{A}, M_{B}, (C_0, m):\, M_A,M_B,C_0\in SL(2),\, \eqref{eq:MAM0MB}, \eqref{eq:n_moncons}\right\rbrace/\sim ,
\end{align}
where $/\sim$ means that we identify monodromy representations related by an overall conjugation.
The space $\mathcal{M}_{1,1}$ has dimension 4 and has coordinates $m,a,\nu,c$, where $c$ parametrizes the freedom of sending $C_0\mapsto C_0e^{c\sigma_3}$ in \eqref{eq:MAM0MB} without changing $M_0$.

The main technique used to obtain our result in \cite{DMDG2020} was the pants decomposition along the $A$-cycle. Namely, the one-point torus can be decomposed into a trinion which 
is a once-punctured cylinder with the identification $z\sim z+1$. It is conformally equivalent to a three-punctured sphere under the map $z\to e^{2\pi iz}$, and the solution of the corresponding linear system serves as an approximation for that on the torus.
\begin{figure}[H]
\begin{center}
\begin{subfigure}{.45\textwidth}
\centering
\begin{tikzpicture}[scale=1.5]
\draw[thick,decoration={markings, mark=at position 0.25 with {\arrow{>}}}, postaction={decorate}] (0,0) circle [x radius=0.5, y radius =0.2];
\draw(-0.5,0)  to[out=270,in=90] (-1.5,-1.5) to[out=-90,in=-90] (1.5,-1.5);
\draw(0.5,0) to[out=270,in=90] (1.5,-1.5);

\draw(-0.7,-1.4) to[out= -30,in=210] (0.7,-1.4);
\draw(-0.55,-1.469) to[out= 30,in=-210] (0.55,-1.469);

\fill[red!30!white] (-0.2,-2.375) to[out=135,in=225] (-0.2,-1.605)
to (0.2,-1.6) to[out=190,in=170] (0.2,-2.373) --cycle;

\draw[dashed,color=black!60!white](-0.2,-2.375) to[out=135,in=225] (-0.2,-1.605)
to (0.2,-1.6) to[out=190,in=170] (0.2,-2.373) --cycle;

\fill[red](-0.2,-2.375) to[out=0,in=0] (-0.2,-1.605)
to (0.2,-1.6) to[out=10,in=-10] (0.2,-2.373) --cycle;

\draw(-0.2,-2.375) to[out=0,in=0] (-0.2,-1.605)
to (0.2,-1.6) to[out=10,in=-10] (0.2,-2.373) --cycle;
\end{tikzpicture}
\caption{One-puncture torus}
\label{fig:Torus}
\end{subfigure}
\hfill
\begin{subfigure}{.45\textwidth}
\centering
\begin{tikzpicture}[scale=1.5]
\draw[thick,decoration={markings, mark=at position 0.25 with {\arrow{>}}}, postaction={decorate}](0,0) circle[x radius=0.5, y radius =0.2];
\draw[red,thick,decoration={markings, mark=at position 0.25 with {\arrow{>}}}, postaction={decorate}] (-1,-2) circle[x radius=0.5, y radius =0.2];
\draw[red,thick,decoration={markings, mark=at position 0.25 with {\arrow{<}}}, postaction={decorate}] (1,-2) circle[x radius=0.5, y radius =0.2];

\draw(-0.5,0)  to[out=270,in=90] (-1.5,-2);
\draw(0.5,0) to[out=270,in=90] (1.5,-2);
\draw(-0.5,-2) to[out=90,in=180] (0,-1.5) to[out=0, in=90] (0.5,-2);
\node at ($(0,-0.7)$) {\Large $\mathscr{T}$};

\end{tikzpicture}
\caption{Trinion}
\label{fig:Trinion}
\end{subfigure}
\end{center}
\vspace{-0.5cm}
\caption{$A$-pants decomposition}
\label{fig:TorusTrinion}
\end{figure}
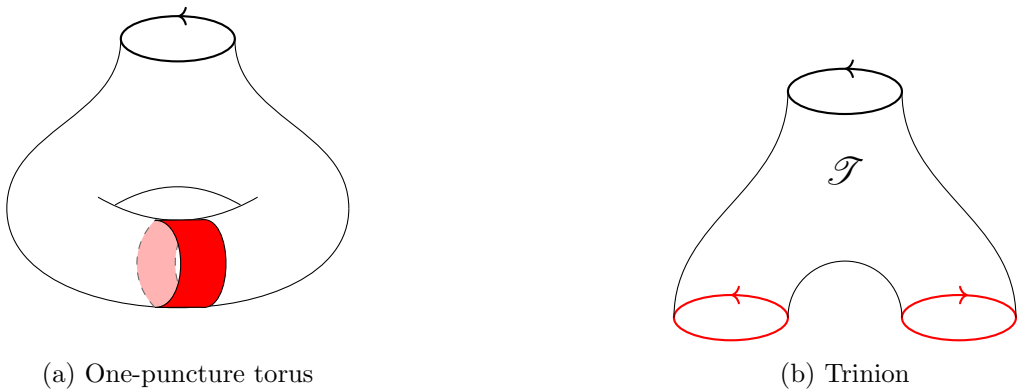

The three-point problem associated to this trinion $\mathscr{T}$ reads 
\begin{align}
    \partial_z Y_{3pt}^{A}(z)=Y_{3pt}^{A}(z) L_{3pt}^{A}(z), && L_{3pt}(z)=-2\pi i L_{-}^{A}-2\pi i\frac{L_0^{A}}{1-e^{2\pi iz}},\label{eq:3ptk}
\end{align}
where the residue matrices are  diagonalizable:
\begin{align}\label{eq:AGmGinv3pt}
    L_{-}^{A}=-(G_{-}^{A})^{-1} a \sigma_3 G_{-}^{A}, && L_0^{A}=(G_0^{A})^{-1} m \sigma_3 G_0^{A},
\end{align} 
\begin{align}\label{def:GA+}
        L_{+}^{A} = - L_{-}^{A}- L_{0}^{A} = (G_{+}^{A})^{-1} a \sigma_3 G_{+}^{A},
\end{align}
and the notation $G_{\bullet}^{A}$ indicates the diagonalization matrices of the three-point problem obtained through a cut along the $A$-cycle of the torus. The local solution on the trinion \(\Phi_{3pt}(z)\) is such that the ratio
\begin{equation*}
\Phi_{3pt}(z)^{-1}\Phi(z)
\end{equation*}
is regular and single-valued around \(z=0\). The solution $\Phi_{3pt}(z)$ therefore approximates the analytic behavior of \(\Phi(z)\) near $z=0$.

In \cite{DMDG2020}, we showed that the tau function \eqref{def:Tau} can be written as a Fredholm determinant that is a function of the modular parameter $\tau$ and the monodromy data. The minor expansion of this Fredholm determinant expresses the tau function as a discrete Fourier transform (more precisely, a Zak transform) of $c=1$ conformal blocks on the torus with one puncture, thereby providing a proof of the results in \cite{Bonelli2019}.
Using the Fredholm determinant representation of the tau function in \cite{DMDG2020}, in \cite{DMDG2022} we obtained the closed form expression for the logarithmic derivative of the tau function with respect to the modular parameter $\tau$ as well as the monodromy data. Let us recall our main result. Since these results follow from studying the pants decomposition along the $A$-cycle, we will denote this tau function by $\T^{A}$.
\begin{theorem}[\cite{DMDG2022}] The {\it total} derivative of the tau function $\T^A$ with respect to the time $\tau$ and the monodromy coordinates $ (a,\,\nu,\, m)$ defined in \eqref{eq:MAM0MB} is given by
\begin{align}\label{this}
    \dd\log\T^{A} = \omega-\omega_{3pt}^{A}, && \dd:= \dd\tau \,\frac{\partial}{\partial 
    \tau} + \dd_{\mathcal{M}},\quad \dd_{\mathcal{M}}:=\dd a\frac{\partial}{\partial a}+\dd\nu\frac{\partial}{\partial \nu}+\dd m\frac{\partial}{\partial m},
\end{align}
where
\begin{align}
    &\omega = 2P\dd_\cM Q+H\frac{\dd\tau}{2\pi i}+\tr\left(m\sigma_3\dd_\cM G\,G^{-1} \right) \mathop{=}^{\eqref{eq:GG0g}} 2 P\dd Q + m \dd_{\mathcal{M}} g + \frac{1}{2\pi i} H \dd\tau, \label{def:omega}\\
    &\omega_{3pt}^{A} = -\tr \left(-{a}\sigma_3 \dd G_{-}^{A} \left(G_{-}^{A}\right)^{-1} +  {a} \sigma_3 \dd G_{+}^{A}\left(G_{+}^{A}\right)^{-1} +  {m} \sigma_3 \dd G_0^{A} \left(G_0^{A}\right)^{-1}\right).\label{def:omega3A}
\end{align}
The Hamiltonian $H$ is defined in \eqref{def:Ham}, the diagonalization matrices $G, G_0^A, G_{+}^{A}, G_{-}^{A}$ are defined in \eqref{def:G}, \eqref{eq:AGmGinv3pt}, \eqref{def:GA+} respectively.
\end{theorem}

Closed expressions like the one above are useful, among other things, in obtaining the ratio of the  tau functions normalised at different critical points, in our case $\tau\to 0$ and $\tau\to i\infty$. This ratio is known as the {\it connection constant}  \cite{ILP2016}. 
The goal of this paper is to obtain the connection constant for the tau function on the one-point torus \eqref{eq:NAECM}, and consequently obtain the modular properties of associated conformal blocks.

\subsubsection*{The modular connection constant}

Equation \eqref{this} is conducive to asymptotic analysis at $\tau\to i\infty$. 
In order to study the behaviour of the tau function for $\tau \to 0$ we need to study instead the pants decomposition obtained by cutting along the $B$-cycle. As we will see below, this amounts to analysing the modular transformation $\tau \to -1/\tau$. We call this the  $B$-pants decomposition and denote the tau function normalised in this way by $\T^{B}$. The relevant monodromy coordinates are denoted by $(\at,\nut,m)$.
Note that $\omega$ in \eqref{this} does not depend on the choice of pants decomposition: the normalisation of the tau function is entirely contained in the term $\omega_{3pt}^A$. $\T^B$ then satisfies 
\begin{align}
    \dd\log \T^{B} = \omega - \omega_{3pt}^{B}.
\end{align}
Therefore, the difference of logarithmic derivatives
\begin{align}\label{def:Upsilon}
    \dd\log \Upsilon_S:= 
     \dd\log \T^{A} - \dd\log \T^{B} 
    = - \omega^{A}_{3pt}+ {\omega}^B_{3pt},
\end{align}
is a 1-form on $T^*\mathcal{M}_{1,1}$ only, independent of the modular parameter $\tau$, and defines the modular connection constant $\Upsilon_S$, where $S$ denotes the $S$-transformation of $SL(2,\mathbb{Z})$.

Our first main result is the integration of \eqref{def:Upsilon}, leading to Theorem \ref{eq:thm:ConnCon}, namely the following explicit formula for the moduar connection constant $\Upsilon_S$:

\begin{align}
    \Upsilon_S(a,\at) &=  \frac{G(1+2a)G(1-2a)}{G(1-m+2a)G(1-m-2a)}\frac{G(1-m+2\ta)G(1-m-2\ta)}{G(1+2\ta)G(1-2\ta) }\cdot \widehat{\Upsilon}_S(a,\at), \label{eq:ConnConstIntro}
\end{align}
\begin{equation}
    \widehat{\Upsilon}_S:=e^{i\,\ta\, \tnu}\frac{\widehat{G}\left(a-m/2+\tnu/(4\pi)\right)\widehat{G}\left(a-m/2-\tnu/(4\pi)\right)}{\widehat{G}\left(a+m/2+\tnu/(4\pi)\right)\widehat{G}\left(a+m/2-\tnu/(4\pi)\right)}\frac{(2\pi)^m\widehat{G}(m)}{e^{i\pi m^2/2}}
\end{equation}

    where $G$ is Barnes' $G$-function, and $\widehat{G}(x):=\frac{G(1+x)}{G(1-x)} $.

\subsubsection*{Modular kernel, conformal blocks and generating functions}

The tau function satisfying \eqref{this} was shown to be related to $c=1$ conformal blocks $\mathcal{B}(a,m,\tau)$ through genus one generalization of the Kyiv formula \cite{desiraju2022painleve}:
\begin{equation}\label{eq:KyivIntro}
\begin{split}
	\frac{1}{\eta(\tau)}\T(a,\nu,m,\tau)&\theta_1(Q(\tau)+\rho|\tau)\theta_1(Q(\tau)-\rho|\tau)\\
    &=\sum_{n,k \in \mathbb{Z}}e^{\frac{i\nu n}{2}}e^{2\pi i\tau\left(k+\frac{n}{2}+\frac{1}{2} \right)^2}e^{4\pi i \left(k+\frac{n}{2}+\frac{1}{2} \right)\left(\rho+\frac{1}{2} \right)}\cB\left(a+\frac{n}{2},m,\tau \right),
    \end{split}
\end{equation}
where $\rho$ is an arbitrary auxiliary parameter, and $\eta(\tau)$ is Dedekind's eta function. Conformal blocks, contrary to tau functions, do not have connection constants, but rather connection \textit{kernels}. The kernel for a modular transformation, known as the modular kernel $S(a,\at)$, relates conformal blocks related by an exchange of $A$- and $B$-cycle through
\begin{equation}\label{eq:SkerneldefIntro}
	\cB(a,m,\tau)=(e^{-i\pi}\tau)^{-\Delta(m)}\int_{-\infty+i\e}^{\infty+i\e}\dd\at\,S(a,\at)\cBt(\at,m,\taut),
\end{equation}
where $\Delta(m)$ is the conformal weight of the primary field inserted at the origin. The factor $(e^{-i\pi}\tau)^{-\Delta(m)}$ takes into account the tensorial transformation properties of conformal blocks, although it is often omitted in recent literature on modular kernels\footnote{The general modular transformation given by the matrix \({\begin{pmatrix}c_1&c_2\\c_3&c_4\end{pmatrix}}\) corresponds to the prefactor \({(c_3\tau+c_4)^{-\Delta(m)}}\).
Notice that here we have the action of the full \(SL(2,\mathbb{Z})\) instead of the modular group \(PSL(2,\mathbb{Z})\).
In particular, \(S^2=-\mathbb{1}\) acts non-trivially.
We also need to choose a representative for \(S\), and we fix it as \(S=\begin{pmatrix}0 & 1\\ -1& 0\end{pmatrix}\).
}.

Determining the modular kernel for $c=1$ conformal blocks is still, to our knowledge, an open problem in two-dimensional conformal field theory. Equation \eqref{eq:KyivIntro}, together with the explicit form of the modular connection constant \eqref{eq:ConnConstIntro}, allows us to give the following explicit expression in Theorem \ref{thm:c1ker}:
    \begin{equation}\label{eq:c1modular}
        S(a,\at)=\frac{\sqrt{2}}{4\pi}\frac{\partial\nu}{\partial\at}\,\widehat{\Upsilon}_S(a,\at).
    \end{equation}
In Section \ref{sec:cinfty}, we conclude by computing the modular kernel for $c\rightarrow\infty$ conformal blocks, and relate it to the connection constant. This also gives a new exact relation between the modular kernels for $c=1$ and $c=\infty$ Virasoro conformal blocks.

Central to proving these results is Proposition \ref{prop:GenFn}, where we prove that both the modular connection constant and the $c\rightarrow\infty$ modular kernel are naturally related to generating functions on the character variety. This is the connection constant counterpart of the relation between isomonodromic tau functions and symplectic geometry of flat connections \cite{BertolaKorotkin2019,DMDG2022}.

\section{Pants decompositions and associated linear systems}\label{sec:pants}

The main goal of this section is to study the trinions obtained through the $A$- and $B$-pants decompositions in detail. We then obtain the explicit form of the diagonalization matrices of the linear systems associated to these trinions in terms of the monodromy data, and obtain the expressions for $\omega_{3pt}^A$ and $\omega_{3pt}^B$, that will be used to compute connection constant $\Upsilon_S$ in the next section. 

\subsection{Linear problem of the trinion from the $A$-pants decomposition.}
In this section we aim to obtain the expression for $\omega_{3pt}^A$ starting from the linear system \eqref{eq:3ptk} on the three-punctured sphere derived through pants decomposition along the $A$-cycle. First, let us recall the following:

\begin{proposition} The behaviour of the solution of NAECM model \eqref{eq:NAECM} $Q(\tau)$ as $\tau\to i\infty$ is given by
\begin{equation}\label{eq:Qetanu}
Q\sim a\tau+\frac{\nu}{4\pi}+\frac{i}{2\pi}\log \frac{\Gamma(2a)\Gamma(1-2a-m)}{\Gamma(1-2a)\Gamma(2a-m)}=a\tau+\eta,\quad e^{2\pi i\eta}:=\frac{\Gamma(1-2a)\Gamma(2a-m)}{\Gamma(2a)\Gamma(1-2a-m)} e^{\frac{i\nu}{2}},
\end{equation}
where $a,\,\nu\in\mathbb{C}$ are monodromy coordinates\footnote{We use the coordinates $\nu$ and $\eta$ interchangeably in this paper.} on the character variety, as in \eqref{eq:MAM0MB}.
\end{proposition}
\begin{proof}
    See \cite[Appendix D]{Bonelli2019}. 
\end{proof}
The local system \eqref{eq:3ptk} is obtained by studying the behaviour of the Lax matrix at $\tau\to i \infty$. The asymptotics of \(x(u,z|\tau)\) reads
\begin{equation}\label{asympxiinf}
x(u,z|\tau)\sim \frac{2\pi i}{e^{2\pi i u}-1}-\frac{2\pi i}{e^{2\pi i z}-1},
\end{equation}
and so, in the same limit, assuming that the monodromy exponent $a$ has small positive real part:
\begin{equation}\label{def:LA3pt}
L(z)\sim 
2\pi i\begin{pmatrix}
a & -\frac{m }{e^{2\pi iz}-1}\\
-\frac{m e^{2\pi i z}}{e^{2\pi i z}-1} & -a
\end{pmatrix}
=:L_{3pt}^{A}(z).
\end{equation}
This is the Lax matrix of the three-point problem associated to the $A$-pants decomposition, which has singularities at $z\rightarrow\pm i\infty$, and $z=0$. In order to compute $\omega_{3pt}^{A}$ given by \eqref{def:omega3A}, all we need is to obtain the matrices $G_{\pm}^{A}, G_0^{A}$ defined in \eqref{eq:AGmGinv3pt}, \eqref{def:GA+} that diagonalize the residue matrices of the linear problem above. 

\begin{proposition}\label{Prop:GAs}
    The matrices $G_{\pm}^{A}, G_{0}^{A}$ read
\begin{align}
    G_{-}^A =
\begin{pmatrix}
e^{i\delta_{-i\infty}} & \\
0 & e^{-i\delta_{-i\infty}}
\end{pmatrix}
\begin{pmatrix}
1 & 0\\
\frac{m}{2a} & 1
\end{pmatrix}, && G_0^A &=
\begin{pmatrix}
e^{i\delta_0} & 0\\
0 & e^{-i\delta_0}
\end{pmatrix}
\begin{pmatrix}
1 & -1\\
1 & 1
\end{pmatrix},
\end{align}
\begin{align}
G_{+}^A =
\begin{pmatrix}
e^{i\nu/2+i\delta_{+i\infty}} \frac{\Gamma(1-2a)\Gamma(2a-m)}{\Gamma(2a)\Gamma(1-2a-m)} & 0\\
0 & e^{-i\nu/2 + i\delta_{+i\infty}}\frac{\Gamma(2a)\Gamma(1-2a-m)}{\Gamma(1-2a)\Gamma(2a-m)}
\end{pmatrix}
\begin{pmatrix}
1 & \frac{m}{2a}\\
0 & 1
\end{pmatrix},
\end{align}
which are determined up to left multiplication by constant diagonal matrices parameterized by $\delta_{-i\infty}$, $\delta_0$, and $\delta_{+i\infty}$ respectively. The parameters $a, \nu, m$ define the monodromy data as in \eqref{eq:MAM0MB}.

\end{proposition}
\begin{proof} We now compute each of the matrices above. Note that the fundamental matrix solution of the linear system
\begin{equation}
\frac{dY_{3pt}^{A}(z)}{dz}=Y_{3pt}^{A}(z)L_{3pt}^{A}(z)
\end{equation}
with $L_{3pt}^{A}(z)$ defined in \eqref{def:LA3pt} is given by hypergeometric functions:
\begin{multline}\footnotesize
Y_{3pt}^{A}(z)=
X_0\,(1-e^{-2\pi i z})^m
\begin{pmatrix}
(-e^{-2\pi i z})^{-a} & 0\\
0 & (-e^{-2\pi iz})^{a}
\end{pmatrix}\\
\times\begin{pmatrix}
_{2}F_1(m,1+m-2a,1-2a,e^{-2\pi iz}) & \frac{-m e^{-2\pi i z}}{2a-1}\, {}_2F_1(1+m,1+m-2a,2-2a,e^{-2\pi iz}) \\
 \frac{m}{2a}\, _{2}F_1(1+m,m+2a,1+2a,e^{-2\pi iz}) & {}_2F_1(m,m+2a,2a,e^{-2\pi i z})
\end{pmatrix},\label{eq:ApY0}
\end{multline}
where \(X_0\) is a normalization matrix that we will specify below. We now study the asymptotics of the above equation to obtain the matrices $G_{+}^{A}, G_{-}^{A}, G_0^{A}$ in the following way.

\paragraph{1. Computation of \(G_{-}^{A}\):}

 To begin with, in the limit $z\to -i \infty$, the solution \eqref{eq:ApY0} has the behaviour
\begin{equation}
\begin{split} 
\label{eq:miinf}
Y_{3pt}^{A}(z\to-i\infty)&\simeq \widetilde{Y}^{A(-i\infty)}_{3pt}(z)=
X_0\begin{pmatrix}
(-e^{-2\pi i z})^{-a} & 0\\
0 & (-e^{-2\pi iz})^{a}
\end{pmatrix}
\begin{pmatrix}
1 & 0 \\
 \frac{m}{2a}  & 1
\end{pmatrix} \\
& =: 
C_{-}^{A}\begin{pmatrix}
(-e^{-2\pi i z})^{-a} & 0\\
0 & (-e^{-2\pi iz})^{a}
\end{pmatrix}
G_{-}^A.
\end{split}
\end{equation}
The matrices
\begin{align}
C_{-}^A=
X_0
\begin{pmatrix}
e^{-i\delta_{-i\infty}} & 0\\
0 & e^{i\delta_{-i\infty}}
\end{pmatrix}, &&
G_{-}^A=
\begin{pmatrix}
e^{i\delta_{-i\infty}} & \\
0 & e^{-i\delta_{-i\infty}}
\end{pmatrix}
\begin{pmatrix}
1 & 0\\
\frac{m}{2a} & 1
\end{pmatrix},
\end{align}
are defined up to an ambiguity of left/right multiplication by a diagonal matrix we parameterize by $\delta_{-i\infty}$. Furthermore, the $A$-cycle monodromy, defined as the analytic continuation $z\mapsto z+1$, with $z$ in the fundamental region, is 
\begin{equation}\label{eq:MAA}
Y_{3pt}^{A}(z+1)=M_AY_{3pt}^{A}(z)=
X_0\begin{pmatrix}
e^{2\pi ia} & 0\\
0 & e^{-2\pi ia}
\end{pmatrix}X_0^{-1}Y_{3pt}^{A}(z), \quad M_A=X_0\begin{pmatrix}
e^{2\pi ia} & 0\\
0 & e^{-2\pi ia}
\end{pmatrix}X_0^{-1}.
\end{equation}

\paragraph{2. Computation of \(G_{+}^{A}\):}

To find the $B$-cycle monodromy we need to analytically continue the hypergeometric function from $x\to -0$ to $x\to -\infty$, obtained through Kummer's formula \cite[\href{https://dlmf.nist.gov/15.10.E25}{eq. (15.10.25)}]{NIST:DLMF}
\begin{multline}
{}_2F_1(a,b,c,x)=\frac{\Gamma(c)\Gamma(b-a)}{\Gamma(b)\Gamma(c-a)}(-x)^{-a}\, {}_2F_1(a,a-c+1,a-b+1,z^{-1})+\\+
\frac{\Gamma(c)\Gamma(a-b)}{\Gamma(a)\Gamma(c-b)}(-x)^{-b}\, {}_2F_1(b,b-c+1,b-a+1,x^{-1}).
\end{multline}
The analytic continuation of the fundamental matrix of solutions \eqref{eq:ApY0} then reads
\begin{multline}
Y_{3pt}^{A}(z)=(1-e^{2\pi iz})^mX_0
\begin{pmatrix}
\frac{\Gamma(1-2a)^2}{\Gamma(1-2a-m)\Gamma(1-2a+m)} & \frac{(2a-1)\Gamma(1-2a)\Gamma(-1+2a)}{m\Gamma(-m)\Gamma(m)}\\
\frac{m \Gamma(1-2a)\Gamma(1+2a)}{2a \Gamma(1-m)\Gamma(1+m)} & \frac{(2a-1)\Gamma(-1+2a)\Gamma(1+2a)}{2a \Gamma(2a-m)\Gamma(2a+m)}
\end{pmatrix}
\begin{pmatrix}
(-e^{-2\pi i z})^{-a} & 0\\
0 & (-e^{-2\pi iz})^{a}
\end{pmatrix}
\times\\\times
\begin{pmatrix}
{}_2F_1(m,2a+m,2a,e^{2\pi iz}) & \frac{m}{2a}\, {}_2F_1(1+m,2a+m,1+2a,e^{2\pi iz})\\
\frac{m e^{2\pi iz}}{1-2a}\, {}_2F_1(1+m,1-2a+m,2-2a,e^{2\pi iz}) &  {}_2F_1(m,1-2a+m,1-2a,e^{2\pi iz})
\end{pmatrix},
\end{multline}
with leading asymptotics as $z\rightarrow+i\infty$ given by
\begin{multline}
Y_{3pt}^{A}(z)\sim
X_0
\begin{pmatrix}
\frac{\Gamma(1-2a)^2}{\Gamma(1-2a-m)\Gamma(1-2a+m)} & \frac{(2a-1)\Gamma(1-2a)\Gamma(-1+2a)}{m\Gamma(-m)\Gamma(m)}\\
\frac{m \Gamma(1-2a)\Gamma(1+2a)}{2a \Gamma(1-m)\Gamma(1+m)} & \frac{(2a-1)\Gamma(-1+2a)\Gamma(1+2a)}{2a \Gamma(2a-m)\Gamma(2a+m)}
\end{pmatrix}
\times\\\times
\begin{pmatrix}
(-e^{-2\pi i z})^{-a} & 0\\
0 & (-e^{-2\pi iz})^{a}
\end{pmatrix}
\begin{pmatrix}
1 & \frac{m}{2a} \\
0 & 1
\end{pmatrix}.
\end{multline}
Rewriting the above expression as
\begin{multline}
\label{eq:piinf}
Y_{3pt}^{A}(z)\sim
X_0
\begin{pmatrix}
\frac{\Gamma(1-2a)^2}{\Gamma(1-2a-m)\Gamma(1-2a+m)} & \frac{(2a-1)\Gamma(1-2a)\Gamma(-1+2a)}{m\Gamma(-m)\Gamma(m)}\\
\frac{m \Gamma(1-2a)\Gamma(1+2a)}{2a \Gamma(1-m)\Gamma(1+m)} & \frac{(2a-1)\Gamma(-1+2a)\Gamma(1+2a)}{2a \Gamma(2a-m)\Gamma(2a+m)}
\end{pmatrix}
\begin{pmatrix}
e^{-2\pi i\eta} & 0\\
0 & e^{2\pi i\eta}
\end{pmatrix}
X_0^{-1}
\widetilde{Y}_{3pt}^{A(+i\infty)}(z),
\end{multline}
where
\begin{equation}
\widetilde{Y}_{3pt}^{A(+i\infty)}(z):=X_0
\begin{pmatrix}
(-e^{-2\pi i z})^{-a} & 0\\
0 & (-e^{-2\pi iz})^{a}
\end{pmatrix}
\begin{pmatrix}
e^{2\pi i\eta} & 0\\
0 & e^{-2\pi i\eta}
\end{pmatrix}
\begin{pmatrix}
1 & \frac{m}{2a} \\
0 & 1
\end{pmatrix}.
\end{equation}

By comparing \eqref{eq:miinf} and \eqref{eq:piinf} order by order in $e^{2\pi i\tau a}$ we can observe that, in the regime $\tau\rightarrow i\infty$, \(z=-\tau/2+\phi\), \(\phi\in \mathbb{R}\), where $Q(\tau)$ is approximated by equation \eqref{eq:Qetanu}, the following asymptotic identity holds:

\begin{align}
&\widetilde{Y}_{3pt}^{A(+i\infty)}(z+\tau) e^{-2\pi i Q\sigma_3} \left( \widetilde{Y}_{3pt}^{A(-i\infty)}(z) \right)^{-1} \nonumber\\
&=
X_0
\begin{pmatrix}1 & \frac{m}{2a} (-e^{-2\pi i(z+\tau)})^{-2a}e^{4\pi i\eta} \\0 & 1\end{pmatrix}
\begin{pmatrix}1 & 0 \\ -\frac{m}{2a} (-e^{-2\pi iz})^{2a} & 1\end{pmatrix} X_0^{-1} \nonumber\\
&=
X_0
\begin{pmatrix}1 & \frac{m}{2a} (-e^{-2\pi i \phi})^{-2a}e^{4\pi i\eta} e^{2\pi i\tau a}\\0 & 1\end{pmatrix}
\begin{pmatrix}1 & 0 \\ -\frac{m}{2a} (-e^{-2\pi i\phi})^{2a} e^{2\pi i \tau a} & 1\end{pmatrix} X_0^{-1}\nonumber\\
&=\mathbb{1}+O(e^{2\pi i a\tau}).
\end{align}
Therefore,
\begin{equation}
Y_{3pt}^{A}(z+\tau)\sim M_BY_{3pt}^{A}(z)e^{2\pi iQ\sigma_3 },
\end{equation}
where 
\begin{align}
M_B&=X_0
\begin{pmatrix}
\frac{\Gamma(1-2a)^2}{\Gamma(1-2a-m)\Gamma(1-2a+m)} & \frac{(2a-1)\Gamma(1-2a)\Gamma(-1+2a)}{m\Gamma(-m)\Gamma(m)}\\
\frac{m \Gamma(1-2a)\Gamma(1+2a)}{2a \Gamma(1-m)\Gamma(1+m)} & \frac{(2a-1)\Gamma(-1+2a)\Gamma(1+2a)}{2a \Gamma(2a-m)\Gamma(2a+m)}
\end{pmatrix}
\begin{pmatrix}
e^{-2\pi i\eta} & 0\\
0 & e^{2\pi i\eta}
\end{pmatrix}
X_0^{-1}.\\
&\mathop{=}^{\eqref{eq:Qetanu}}X_0 
\begin{pmatrix}
\frac{\sin \pi(2a-m)}{\sin 2\pi a} & - \frac{\Gamma(1-2a)^2\Gamma(2a-m)\sin \pi m}{\pi \Gamma(1-2a-m)}\\
\frac{\Gamma(2a)^2\Gamma(1-2a-m)\sin \pi m}{\pi \Gamma(2a-m)} & \frac{\sin \pi(2a+m)}{\sin 2\pi a}
\end{pmatrix}
\begin{pmatrix}
e^{-i\nu/2} & 0\\
0 & e^{i\nu/2}
\end{pmatrix}X_0^{-1}.\label{eq:MBA}
\end{align}
In conclusion, the asymptotic behaviour of the solution in terms of \(C\) and \(G\) matrices is written as
\begin{equation}
Y_{3pt}^{A}(z)=C_{+}^A
\begin{pmatrix}
(-e^{-2\pi i z})^{-a} & 0\\
0 & (-e^{-2\pi iz})^{a}
\end{pmatrix} G_{+}^A,
\end{equation}
where
\begin{align}
C_{+}^A &=M_B X_0
\begin{pmatrix}
e^{-i\delta_{+i\infty}} & 0\\
0 & e^{i\delta_{+i\infty}}
\end{pmatrix},\\
G_{+}^A &=
\begin{pmatrix}
e^{i\nu/2+i\delta_{+i\infty}} \frac{\Gamma(1-2a)\Gamma(2a-m)}{\Gamma(2a)\Gamma(1-2a-m)} & 0\\
0 & e^{-i\nu/2 + i\delta_{+i\infty}}\frac{\Gamma(2a)\Gamma(1-2a-m)}{\Gamma(1-2a)\Gamma(2a-m)}
\end{pmatrix}
\begin{pmatrix}
1 & \frac{m}{2a}\\
0 & 1
\end{pmatrix}.
\end{align}
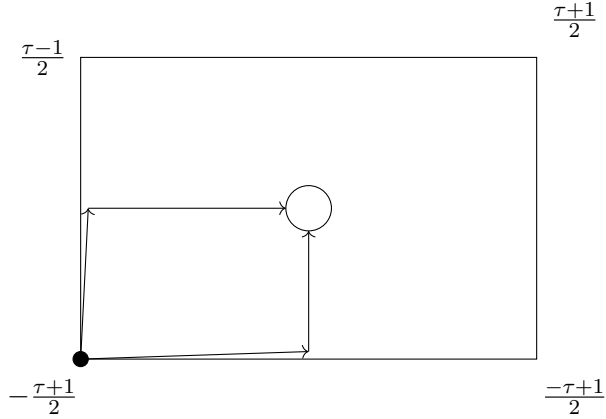
\begin{figure}[H]
\begin{center}
\begin{tikzpicture}[scale=1]
\draw[](-3,-2)--(-3,2)--(3,2)--(3,-2)--cycle;
\draw[radius=0.3](0,0) circle;
\draw[radius=0.1,fill=black](-3,-2) circle;
\draw[->](-3,-2)--(-2.9,0);\draw[->](-2.9,0)--(-0.3,0);
\draw[->](-3,-2)--(0,-1.9);\draw[->](0,-1.9)--(0,-0.3);
\node  at (-3.5,-2.5) {\(-\frac{\tau+1}{2}\)};
\node  at (-3.5,2) {\(\frac{\tau-1}{2}\)};
\node  at (3.5,-2.5) {\(\frac{-\tau+1}{2}\)};
\node  at (3.5,2.5) {\(\frac{\tau+1}{2}\)};
\end{tikzpicture}
\end{center}
\caption{\label{fig:monodromies} Monodromies and analytic continuations}
\end{figure}

\paragraph{3. Computation of \(G_{0}^{A}\):}
As the next step, we compute the asymptotics of $Y_{3pt}^{A}(z)$ around \(z=0\). To do this, we first need to perform the analytic continuation to the vicinity of \(z=-\frac{\tau}{2}\), and then use the formulae \cite[\href{https://dlmf.nist.gov/15.10.E21}{eq. (15.10.21)}]{NIST:DLMF} for the analytic continuation of \({}_2F_1(\cdot, x)\) from \(0+\epsilon\) to \(1-\epsilon\) (see fig.~\ref{fig:monodromies}):
\begin{align}
{}_2F_1(a,b,c,x)&=\frac{\Gamma(c)\Gamma(c-a-b)}{\Gamma(c-a)\Gamma(c-b)}{}_2F_1(a,b,1+a+b-c,1-x)+\nonumber\\ 
&+ \frac{\Gamma(a+b-c)\Gamma(c)}{\Gamma(a)\Gamma(b)}(1-x)^{c-a-b}{}_2F_1(c-a,c-b,c-a-b+1,1-x),\\
{}_2F_1(a,b,c,x)&=(1-x)^{c-a-b}{}_2F_1(c-a,c-b,c,x).
\end{align}
With the above expressions, we obtain that
\begin{multline}
Y_{3pt}^{A}(z)=X_0(e^{-2\pi iz})^{-a}
\begin{pmatrix}
\frac{e^{\pi i a}\Gamma(1-2a)\Gamma(-2m)}{\Gamma(1-2a-m)\Gamma(-m)} & -\frac{e^{\pi i a}\Gamma(1-2a)\Gamma(2m)}{\Gamma(m)\Gamma(1-2a+m)}\\
-\frac{e^{-\pi i a}\Gamma(2a)\Gamma(-2m)}{\Gamma(2a-m)\Gamma(-m)} & -\frac{e^{-\pi ia} \Gamma(2a) \Gamma(2m)}{\Gamma(m)\Gamma(2a+m)}
\end{pmatrix}
\begin{pmatrix}
(1-e^{-2\pi i z})^m & 0\\
0 & (1-e^{-2\pi i z})^{-m}
\end{pmatrix}\times\\\times
\begin{pmatrix}
{}_2F_1(m,1-2a+m,1+2m,1-e^{2\pi iz}) & -e^{-2\pi iz}{}_2F_1(1+m,1-2a+m,1+2m,1-e^{-2\pi i z})\\
{}_2F_1(-m,1-2a-m,1-2m,1-e^{-2\pi i z}) & e^{-2\pi iz}{}_2F_1(1-m,1-2a-m,1-2m,1-e^{-2\pi iz})
\end{pmatrix},
\end{multline}
which, in the limit $z\to 0$, behaves as
\begin{equation}
Y_{3pt}^{A}(z)\sim X_0 
\begin{pmatrix}
\frac{e^{\pi i a}\Gamma(1-2a)\Gamma(-2m)}{\Gamma(1-2a-m)\Gamma(-m)} & -\frac{e^{\pi i a}\Gamma(1-2a)\Gamma(2m)}{\Gamma(m)\Gamma(1-2a+m)}\\
-\frac{e^{-\pi i a}\Gamma(2a)\Gamma(-2m)}{\Gamma(2a-m)\Gamma(-m)} & -\frac{e^{-\pi ia} \Gamma(2a) \Gamma(2m)}{\Gamma(m)\Gamma(2a+m)}
\end{pmatrix}
\begin{pmatrix}
(2\pi i z)^m & 0\\
0 & (2\pi iz)^{-m}
\end{pmatrix}
\begin{pmatrix}
1 & -1\\
1 & 1
\end{pmatrix}.
\end{equation}
The branch of \((2\pi i z)^m\) is chosen by noting that \((1-e^{-2\pi iz})^m\) is initially defined for \(z=-\frac{\tau}{2}\), and so, \(\arg (i z)=0\) for \(z=-i0\). The asymptotics at $z\sim 0$ is written in terms of \(C\) and \(G\) matrices as
\begin{equation}
\label{eq:Y3pt0A}
Y_{3pt}^{A}(z)\simeq C_0^A 
\begin{pmatrix}
(2\pi i z)^m & 0 \\
0 & (2\pi i z)^{-m}
\end{pmatrix}
G_0^A,
\end{equation}
where
\begin{align}
C_0^A &=
X_0\begin{pmatrix}
\frac{e^{\pi i a}\Gamma(1-2a)\Gamma(-2m)}{\Gamma(1-2a-m)\Gamma(-m)} & -\frac{e^{\pi i a}\Gamma(1-2a)\Gamma(2m)}{\Gamma(m)\Gamma(1-2a+m)}\\
-\frac{e^{-\pi i a}\Gamma(2a)\Gamma(-2m)}{\Gamma(2a-m)\Gamma(-m)} & -\frac{e^{-\pi ia} \Gamma(2a) \Gamma(2m)}{\Gamma(m)\Gamma(2a+m)}
\end{pmatrix}
\begin{pmatrix}
e^{-i\delta_0} & 0\\
0 & e^{i\delta_0}
\end{pmatrix},\label{eq:C0A}\\
G_0^A &=
\begin{pmatrix}
e^{i\delta_0} & 0\\
0 & e^{-i\delta_0}
\end{pmatrix}
\begin{pmatrix}
1 & -1\\
1 & 1
\end{pmatrix}.
\end{align}
\end{proof}
The above information gives us all we need to compute $\omega_{3pt}^{A}$. Instead of computing $\omega_{3pt}^{A}$ at this point, we will soon see that it will be easier to compute $\omega_{3pt}^{A}- \omega_{3pt}^{B}$ directly.

\subsection{Linear problem of the trinion from $B$-pants decomposition.}
The asymptotic analysis for $\tau\to 0$ can be inferred by repeating the above computations, constructing the local parametrix by a cut along the $B$-cycle of the torus. As $\tau$ is always assumed to lie on the upper half-plane, this limit will be approached along a positive imaginary direction, and sometimes denoted by $\tau\rightarrow i0$. Consider the linear problem associated to the trinion: 
\begin{align}
    \partial_z Y_{3pt}^{B}(z)=Y_{3pt}^{B}(z) L_{3pt}^{B}(z), && L_{3pt}^{B}(z)=-2\pi i L_{-}^{B}-2\pi i\frac{L_0^{B}}{1-e^{2\pi iz}},\label{eq:3ptkB}
\end{align}
where, in analogy with \eqref{eq:AGmGinv3pt}, the residue matrices are diagonalizable
\begin{align}\label{eq:BGmGinv3pt}
    L_{-}^{B}=-(G_{-}^{B})^{-1}  \widetilde{ a}\sigma_3\, G_{-}^{B}, && L_0^{B}=(G_0^{B})^{-1} m\sigma_3 \,G_0^{B},
\end{align} 
\begin{align}
        L_{+}^{B} = - L_{-}^{B}- L_{0}^{B} = (G_{+}^{B})^{-1} \widetilde{a} \sigma_3 \, G_{+}^{B}.
\end{align}
We will refer to $(\at,\nut)$ as the dual monodromy coordinates, and the precise relation between the dual monodromy data we call $\left(\widetilde{a}, \widetilde{\nu} \right)$ and the coordinates $(a, \nu)$ will be derived in the next section. The $B$-cycle analogue of the one-form \eqref{def:omega3A} reads
\begin{align}
    \omega_{3pt}^{B} &= -\tr (-{\ta}\sigma_3) \dd G_{-}^{B} \left(G_{-}^{B}\right)^{-1} - \tr {\ta}\sigma_3 \dd G_{+}^{B}\left(G_{+}^{B}\right)^{-1}- \tr {m}\sigma_3 \dd G_0^{B} \left(G_0^{B}\right)^{-1}.\label{def:omega3B}
\end{align}
Using the above relation, we find the following result.
\begin{proposition}
    The asymptotics of solution $Q(\tau)$ in \eqref{eq:NAECM} for $\tau\rightarrow +i0$ reads
\begin{equation}\label{0asymp:Q}
Q(\tau\rightarrow0)\sim \ta-\tau \left( \frac{\tnu}{4\pi}+ \frac{i}{2\pi} \log \frac{\Gamma(2\ta)\Gamma(1-2\ta-m)}{\Gamma(1-2\ta)\Gamma(2\ta-m)} \right)=\ta-\teta\,\tau,
\end{equation}
\begin{align}\label{def:teta}
    e^{2\pi i\teta}:=\frac{\Gamma(1-2\ta)\Gamma(2\ta-m)}{\Gamma(2\ta)\Gamma(1-2\ta-m)} e^{\frac{i\tnu}{2}}.
\end{align}
\end{proposition}
\begin{proof}
Let $\taut:=-1/\tau$. Together with $\Qt(\taut):=-Q(\tau)/\tau$ this is a symmetry of the isomonodromic equation \eqref{eq:NAECM} \cite{manin1996sixth}\footnote{See Section \ref{sec:SDual} for more details.}. The asymptotics as $\taut\rightarrow i\infty$ will then have the same form as \eqref{eq:Qetanu}, with different parameters:
 \begin{align*}
     \widetilde{Q}(\ttau) \sim\ta\, \ttau + \teta = -\left(\ta -\teta \tau\right)/\tau.
 \end{align*}
 Then noting that $\widetilde{\tau} \to +i\infty$ is equivalent to $\tau\to +i0$, and using the relation $\Qt=-Q/\tau$, we obtain the result \eqref{0asymp:Q}. 
\end{proof}

For now \eqref{def:teta} is simply a definition of $\nut$ in terms of the asymptotics of $\Qt$, which will be justified in the following.

As in the previous subsection, we use the behaviour \eqref{0asymp:Q} to study the behaviour of the Lax matrix $L_z$ in \eqref{linear_systemCM}, for which, we begin by computing asymptotics of \(x(u,z|\tau)\) around \(\tau\to +i0\). We start with the following identities of theta function for $\ttau = -1/\tau$  \cite[20.7]{NIST:DLMF}
\begin{align*}
   (- i \tau)^{1/2} \theta_1(z\vert \tau) = -i e^{i\pi z^2 \ttau} \theta_1(z\ttau\vert \ttau), &&  (- i \tau)^{1/2} \theta_1'(0\vert \tau) = - i \ttau \theta_1'(0\vert \ttau),
\end{align*}
and consequently, 
\begin{equation}
x(u,z|\tau)= \ttau e^{-2\pi izu \ttau} x(u\ttau, z\ttau|\ttau).
\end{equation}
Using \eqref{asympxiinf}, we observe that, in the limit $\ttau\to +i\infty$, the RHS of the above expression behaves as
\begin{align*}
 \ttau e^{-2\pi izu \ttau} x(u\ttau, z\ttau|\ttau)   \sim 2\pi i \ttau e^{-2\pi i zu \ttau}\left( \frac1{e^{2\pi iu \ttau}-1}-\frac1{e^{2\pi iz \ttau}-1} \right),
\end{align*}
which implies the following behaviour of $x(u,z\vert \tau)$ for $\tau\to +i0$:
\begin{align}
    x(u,z\vert \tau) \sim \frac{2\pi i}{\tau} e^{2\pi i zu/\tau}\left( \frac1{e^{2\pi iu/\tau}-1}-\frac1{e^{2\pi iz/\tau}-1} \right).
\end{align}
Assuming that \(\Re \ta>0\), we obtain the following behaviour for the Lax matrix $L_z$ \eqref{linear_systemCM} in the limit \(\tau\to +i0\):
\begin{equation}
L_z(z)\sim 
\frac{2\pi i}{\tau}\begin{pmatrix}
-{\teta}\,\tau & \frac{-m e^{-4\pi i z(\ta-{\teta}\tau)/\tau}e^{2\pi i z/\tau}}{e^{2\pi i z/\tau}-1}  \\
\frac{-m e^{4\pi i z({\ta}-{\teta}\tau)/\tau}}{e^{2\pi i z/\tau}-1} & {\teta}\,\tau
\end{pmatrix}=:{L}_{3pt}^{B}(z).
\end{equation}
The above matrix describes the linear system on a 3-point sphere obtained by cutting the torus along the $B$-cycle in \eqref{eq:3ptkB}. We now compute its diagonalization matrices \eqref{eq:BGmGinv3pt} that are needed to compute the one-form $\omega_{3pt}^{B}$ \eqref{def:omega3B}.
\begin{proposition}\label{Prop:GBs}
    The matrices $G_{\pm}^{B}, G_{0}^{B}$ are given by
\begin{align}
    G_{-}^B =
\begin{pmatrix}
e^{i\widetilde{\delta}_{-\infty}} & 0\\
0 & e^{-i\widetilde{\delta}_{-\infty}}
\end{pmatrix}
\begin{pmatrix}
1 & 0\\
\frac{m}{2\widetilde{a}} & 1
\end{pmatrix}, && G_0^B &=
\begin{pmatrix}
e^{i\widetilde{\delta}_0} & 0\\
0 & e^{-i\widetilde{\delta}_0}
\end{pmatrix}
\begin{pmatrix}
1 & -1\\
1 & 1
\end{pmatrix},
\end{align}
\begin{align}
G_{+}^B =
\begin{pmatrix}
e^{i\tnu/2+i\widetilde{\delta}_{+\infty}} \frac{\Gamma(1-2\ta)\Gamma(2\ta-m)}{\Gamma(2\ta)\Gamma(1-2\ta-m)} & 0\\
0 & e^{-i\tnu/2-i\widetilde{\delta}_{+\infty}} \frac{\Gamma(2\ta)\Gamma(1-2\ta-m)}{\Gamma(1-2\ta)\Gamma(2\ta-m)}
\end{pmatrix}
\begin{pmatrix}
1 & \frac{m}{2\ta}\\
0 & 1
\end{pmatrix},
\end{align}
which are all determined up to left multiplication by constant diagonal matrices parameterized by $\widetilde{\delta}_{-\infty}$, $\widetilde{\delta}_0$, and $\widetilde{\delta}_{+\infty}$ respectively. 

\end{proposition}

\begin{proof}
We compute each of the matrices above. 

\paragraph{1. Computation of \(G_{-}^{B}\):}

The fundamental solution of the equation \eqref{eq:3ptkB} reads as
\begin{multline}\footnotesize
Y_{3pt}^{B}=
X_{\infty}(1-e^{2\pi i z/\tau})^m
\begin{pmatrix}
(-e^{2\pi iz/\tau})^{-{\ta}} & 0\\
0 & (-e^{2\pi i z/\tau})^{{\ta}}
\end{pmatrix}\times\\
\\\times
\begin{pmatrix}
{}_2F_1(m,1+m-2\ta,1-2\ta,e^{2\pi iz/\tau}) & \frac{-m e^{2\pi iz/\tau}}{2\ta-1}{}_2F_1(1+m,1+m-2\ta,2-2\ta,e^{2\pi iz/\tau})\\
\frac{m}{2\ta} {}_2F_1(1+m,m+2\ta,1+2\ta,e^{2\pi iz/\tau}) & {}_2F_1(m,m+2\ta,2\ta,e^{2\pi iz/\tau})
\end{pmatrix}\times\\\times
\begin{pmatrix}
e^{2\pi iz(\ta-\teta\tau)/\tau} &0 \\
0 & e^{-2\pi iz(\ta-\teta)/\tau}
\end{pmatrix}\\=:
Y_{\infty}^{B}(z)\begin{pmatrix}
e^{2\pi iz(\ta-\teta\tau)/\tau} &0 \\
0 & e^{-2\pi iz(\ta-\teta)/\tau}
\end{pmatrix}.\label{eq:fsY3ptB}
\end{multline}
Sometimes we prefer to work with solutions for the 3-point problem that do not have twists, in this case we can choose \(Y_{\infty}^{B}(z)\).
This solution contains the matrix \(X_{\infty}\) which should be chosen in such a way that its monodromies are consistent with monodromies of \(Y_{3pt}^{A}(z)\) in \eqref{eq:ApY0}.
In order to obtain $G_{-}^{B}$, we begin by computing the $B$-cycle monodromy:
\begin{equation}
Y_{3pt}^{B}(z+\tau)=M_BY_{3pt}^{B}(z)e^{2\pi iQ\sigma_3},
\end{equation}
and so,
\begin{equation}\label{eq:MBB}
M_B=X_{\infty}
\begin{pmatrix}
e^{-2\pi i \ta} & 0\\
0 & e^{2\pi i \ta}
\end{pmatrix}
X_{\infty}^{-1}.
\end{equation}
In the limit $z\to -i \infty$, the solution \eqref{eq:fsY3ptB} is given by
\begin{equation}
{Y}_{3pt}^{B}(z)=C_{-}^{B}
\begin{pmatrix}
(-e^{2\pi iz/\tau})^{-\ta} & 0\\
0 & (-e^{2\pi i z/\tau})^{\ta}
\end{pmatrix}
G_{-}^{B}\begin{pmatrix}
e^{2\pi iz(\ta-\teta\tau)/\tau} &0 \\
0 & e^{-2\pi iz(\ta-\teta)/\tau}
\end{pmatrix},
\end{equation}
where
\begin{align}
C_{-}^{B}=X_{\infty} 
\begin{pmatrix}
e^{-i\widetilde{\delta}_{-i\infty}} & 0\\
0 & e^{+i\widetilde{\delta}_{-i\infty}}
\end{pmatrix},&&
G_{-}^{B}=
\begin{pmatrix}
e^{i\widetilde{\delta}_{-i\infty}} & 0\\
0 & e^{-i\widetilde{\delta}_{-i\infty}}
\end{pmatrix}
\begin{pmatrix}
1 & 0\\
\frac{m}{2\widetilde{a}} & 1
\end{pmatrix},
\end{align}
for an arbitrary parameter $\widetilde{\delta}_{-i\infty}$.

\paragraph{2. Computation of \(G_{+}^{B}\):}

In order to compute $G_{+}^{B}$, we need to study the behaviour of \({Y}_{3pt}^{B}\) near $+i\infty$. This is done by analytic continuation of hypergeometric functions which give us the following expression for $Y_{\infty}^{B}$ in \eqref{eq:fsY3ptB}:
\begin{multline}
{Y}_{\infty}^{B}(z)=X_{\infty} (1-e^{-2\pi iz/\tau})^m
\begin{pmatrix}
\frac{\Gamma(1-2\ta)^2}{\Gamma(1-2\ta-m)\Gamma(1-2\ta+m)} & \frac{(2\ta-1)\Gamma(1-2\ta)\Gamma(2\ta-1)}{m\Gamma(-m)\Gamma(m)}\\
\frac{m \Gamma(1-2\ta)\Gamma(1+2\ta)}{2\ta\Gamma(1-m)\Gamma(1+m)} & \frac{(2\ta-1)\Gamma(2\ta-1)\Gamma(1+2\ta)}{2\ta \Gamma(2\ta-m)\Gamma(2\ta+m)}
\end{pmatrix}
\begin{pmatrix}
(-e^{2\pi i z/\tau})^{-\ta} & 0\\
0 & (-e^{2\pi i z/\tau})^{\ta}
\end{pmatrix}
\\\times
\begin{pmatrix}
{}_2F_1(m,2\ta+m,2\ta,e^{-2\pi iz/\tau}) & \frac{m}{2 \ta} {}_2F_1(1+m,2\ta+m,1+2\ta,e^{-2\pi iz/\tau})\\
\frac{m e^{-2\pi iz/\tau}}{1-2\ta} {}_2F_1(1+m,1-2\ta+m,2-2\ta) & {}_2F_1(m,1-2\ta+m,1-2\ta,e^{-2\pi iz/\tau})
\end{pmatrix}.\label{eq:fsY3ptB2}
\end{multline}
We now compute $A$-cycle monodromy:
\begin{equation}
Y_{\infty}^{B}(z+1)=M_AY_{\infty}^{B}(z).
\end{equation}
Plugging in \eqref{eq:fsY3ptB2}, up to some infinitely small terms which appear in this approximate computation, we see that
\begin{equation}\label{eq:MAB}
M_A=X_{\infty}
\begin{pmatrix}
\frac{\Gamma(1-2\ta)^2}{\Gamma(1-2\ta-m)\Gamma(1-2\ta+m)} & \frac{(2\ta-1)\Gamma(1-2\ta)\Gamma(2\ta-1)}{m\,\Gamma(-m)\Gamma(m)}\\
\frac{m\, \Gamma(1-2\ta)\Gamma(1+2\ta)}{2\ta\,\Gamma(1-m)\Gamma(1+m)} & \frac{(2\ta-1)\Gamma(2\ta-1)\Gamma(1+2\ta)}{2\ta\, \Gamma(2\ta-m)\Gamma(2\ta+m)}
\end{pmatrix}
\begin{pmatrix}
e^{-2\pi i\teta} & 0\\
0 & e^{2\pi i \teta}
\end{pmatrix}X_{\infty}^{-1}.
\end{equation}
Using the definition of \(\teta\) \eqref{def:teta}, the above expression can be rewritten as
\begin{equation}
M_A=
X_{\infty}\begin{pmatrix}
\frac{\sin\pi(2\ta-m)}{\sin 2\pi \ta} & -\frac{\Gamma(1-2\ta)^2\Gamma(2\ta-m)\sin \pi m}{\pi \Gamma(1-2\ta-m)} \\
\frac{\Gamma(2\ta)^2\Gamma(1-2\ta-m)\sin \pi m}{\pi \Gamma(2\ta-m)} & \frac{\sin\pi(2\ta+m)}{\sin 2\pi \ta}
\end{pmatrix}
\begin{pmatrix}
e^{-i\tnu/2} & 0\\
0 & e^{i\tnu/2}
\end{pmatrix}
X_{\infty}^{-1}.
\end{equation}
As the final step, we note that the asymptotic behaviour of $Y_{\infty}^{B}$ in \eqref{eq:fsY3ptB2} for $z\to +i\infty$ reads
\begin{equation}
Y_{\infty}^{B}(z)=C_{+}^{B}
\begin{pmatrix}
(-e^{2\pi iz/\tau})^{-\ta} & 0\\
0 & (-e^{2\pi i z/\tau})^{\ta}
\end{pmatrix}
G_{+}^{B}\begin{pmatrix}
e^{2\pi iz(\ta-\teta\tau)/\tau} &0 \\
0 & e^{-2\pi iz(\ta-\teta\tau)/\tau}
\end{pmatrix},
\end{equation}
where, for arbitrary parameter $\widetilde{\delta}_{+i\infty}$,
\begin{equation}
C_{+}^{B}=M_AX_{\infty} 
\begin{pmatrix}
e^{-i\widetilde{\delta}_{+i\infty}} & 0\\
0 & e^{i\widetilde{\delta}_{+i\infty}}
\end{pmatrix},
\end{equation}
\begin{equation}
G_{+}^{B}=
\begin{pmatrix}
e^{i\tnu/2+i\widetilde{\delta}_{+i\infty}} \frac{\Gamma(1-2\ta)\Gamma(2\ta-m)}{\Gamma(2\ta)\Gamma(1-2\ta-m)} & 0\\
0 & e^{-i\tnu/2-i\widetilde{\delta}_{+i\infty}} \frac{\Gamma(2\ta)\Gamma(1-2\ta-m)}{\Gamma(1-2\ta)\Gamma(2\ta-m)}
\end{pmatrix}
\begin{pmatrix}
1 & \frac{m}{2\widetilde{a}}\\
0 & 1
\end{pmatrix}.
\end{equation}

\paragraph{3. Computation of \(G_{0}^{B}\):}

The diagonalization matrix $G_0^{B}$ is obtained by studying the behaviour of the solution \eqref{eq:fsY3ptB} near $z=0$, which can be achieved through analytic continuation to \(z=0\). To do this, as in the proof of Proposition \ref{Prop:GAs}, we first move from \(z=-\frac{\tau+1}{2}\) to \(z=-\frac{1}{2}\), and then perform analytic continuation of the hypergeometric function. We obtain that
\begin{multline}
{Y}_{3pt}^{B}(z)=
X_{\infty}\begin{pmatrix}
\frac{e^{-\pi i \ta}\Gamma(1-2\ta)\Gamma(-2m)}{\Gamma(1-2\ta-m)\Gamma(-m)} & -\frac{e^{-\pi i\ta}\Gamma(1-2\ta)\Gamma(2m)}{\Gamma(m)\Gamma(1-2\ta+m)}\\
-\frac{e^{\pi i \ta}\Gamma(2\ta)\Gamma(-2m)}{\Gamma(2\ta-m)\Gamma(-m)} & -\frac{e^{\pi i \ta}\Gamma(2\ta)\Gamma(2m)}{\Gamma(m)\Gamma(2\ta+m)}
\end{pmatrix}
(e^{2\pi iz/\tau})^{-\ta}
\begin{pmatrix}
(1-e^{2\pi iz/\tau})^m & 0\\
0 & (1-e^{2\pi iz/\tau})^{-m}
\end{pmatrix}
\\\times
\begin{pmatrix}
{}_2F_1(m,1-2\ta+m,1+2m,1-e^{2\pi iz/\tau}) & -e^{2\pi iz/\tau}{}_2F_1(1+m,1-2\ta+m,1+2m,1-e^{2\pi iz/\tau})\\
{}_2F_1(-m,1-2\ta-m,1-2m,1-e^{2\pi iz/\tau}) & e^{2\pi iz/\tau}{}_2F_1(1-m,1-2\ta-m,1-2m,1-e^{2\pi iz/\tau})
\end{pmatrix}
\\\times
\begin{pmatrix}
e^{2\pi iz(\ta-\teta\tau)/\tau} &0 \\
0 & e^{-2\pi iz(\ta-\teta\tau)/\tau}
\end{pmatrix}.
\end{multline}
The above matrix behaves as follows near $z=0$:
\begin{multline}
Y_{3pt}^{B}(z)=
X_{\infty}\begin{pmatrix}
\frac{e^{-\pi i \ta}\Gamma(1-2\ta)\Gamma(-2m)}{\Gamma(1-2\ta-m)\Gamma(-m)} & -\frac{e^{-\pi i\ta}\Gamma(1-2\ta)\Gamma(2m)}{\Gamma(m)\Gamma(1-2\ta+m)}\\
-\frac{e^{\pi i \ta}\Gamma(2\ta)\Gamma(-2m)}{\Gamma(2\ta-m)\Gamma(-m)} & -\frac{e^{\pi i \ta}\Gamma(2\ta)\Gamma(2m)}{\Gamma(m)\Gamma(2\ta+m)}
\end{pmatrix}
\begin{pmatrix}
(-2\pi iz/\tau)^m & 0\\
0 & (-2\pi iz/\tau)^{-m}
\end{pmatrix}
\begin{pmatrix}
1 & -1\\
1 & 1
\end{pmatrix},
\end{multline}
where \(\arg(-iz/\tau)=0\) for \(z=-0\).  Alternatively we can write the above expression as
\begin{equation}
\label{eq:Y3pt0B}
Y_{3pt}^{B}(z)=C_0^{B}
\begin{pmatrix}
(-2\pi iz/\tau)^m & 0\\
0 & (-2\pi iz/\tau)^{-m}
\end{pmatrix}
{G}_0^{B}
\begin{pmatrix}
e^{2\pi iz(\ta-\teta\tau)/\tau} &0 \\
0 & e^{-2\pi iz(\ta-\teta\tau)/\tau}
\end{pmatrix}
,
\end{equation}
where
\begin{equation}\label{eq:C0B}
C_0^B=
X_{\infty}\begin{pmatrix}
\frac{e^{-\pi i \ta}\Gamma(1-2\ta)\Gamma(-2m)}{\Gamma(1-2\ta-m)\Gamma(-m)} & -\frac{e^{-\pi i\ta}\Gamma(1-2\ta)\Gamma(2m)}{\Gamma(m)\Gamma(1-2\ta+m)}\\
-\frac{e^{\pi i \ta}\Gamma(2\ta)\Gamma(-2m)}{\Gamma(2\ta-m)\Gamma(-m)} & -\frac{e^{\pi i \ta}\Gamma(2\ta)\Gamma(2m)}{\Gamma(m)\Gamma(2\ta+m)}
\end{pmatrix}
\begin{pmatrix}
e^{-i\widetilde{\delta}_0} & 0\\
0 & e^{i\widetilde{\delta}_0}
\end{pmatrix},
\end{equation}
\begin{equation}
G_0^{B}=
\begin{pmatrix}
e^{i\widetilde{\delta}_0} & 0\\
0 & e^{-i\widetilde{\delta}_0}
\end{pmatrix}
\begin{pmatrix}
1 & -1\\
1 & 1
\end{pmatrix},
\end{equation}
for an arbitrary parameter $\widetilde{\delta}_{0}$.
\end{proof}

\begin{remark}\label{rmk:minus1}
Expressions for the \(A\)-cycle and \(B\)-cycle parametrices \eqref{eq:Y3pt0A} and \eqref{eq:Y3pt0B} do not have branch cuts in the lower-left part of the fundamental domain\footnote{This is so because two boundaries of this domain are the deformations of the same path for analytic continuation from \(-\frac{\tau+1}{2}\) to 2, see Fig. \ref{fig:monodromies}.}.
This allows us to conclude that for \(\tau\in i \mathbb{R}\), \(\arg(-1/\tau)=i\pi/2\).
We can additionally choose \(\arg(\tau)\in (0,\pi)\) everywhere throughout the paper. This gives us the value of the argument of \((-1)\) in the formula \eqref{eq:Y3pt0B}: \(\arg(-1)=+i\pi\).
\end{remark}

\subsection{Comparison of solutions}

From the above subsections, we have two solutions $Y_{3pt}^{A}$ \eqref{eq:3ptk} and $Y_{3pt}^{B}$ \eqref{eq:3ptkB} that approximate solution of our linear problem $Y$ \eqref{linear_systemCM} in two different limits. Their monodromy data and asymptotics should therefore be related, and we obtain the following result.
\begin{proposition}\label{prop:SMonodromy} These two formulas allow one to express all monodromy data in terms of just \(\ta\) and \(\nu\).
    \begin{align} 
a&=\frac1{4\pi i}\log \frac{\sin(\pi\ta-\pi m/2+\nu/4)\sin(\pi \ta+\pi m/2-\nu/4)}{\sin(\pi \ta-\pi m/2-\nu/4)\sin(\pi \ta+\pi m/2+\nu/4)},\label{rel:a_as_atnu}\\
\tnu&=4\pi a-2i \log \frac{\sin(\pi \ta-\pi m/2-\nu/4)}{\sin(\pi \ta+\pi m/2-\nu/4)}.\label{rel:nut_as_aatnu}
\end{align}
\end{proposition}
\begin{proof}
In order to obtain the precise map between the data $(a, \nu)$ and $(\ta, \tnu)$ respectively, it is sufficient to compare the traces of monodromies.  Comparing \eqref{eq:MAA}, \eqref{eq:MAB}, we get
\begin{equation}\label{rel:MAs}
\tr M_A=e^{2\pi i a}+e^{-2\pi i a}=e^{-i\tnu/2}\frac{\sin\pi(2\ta-m)}{\sin2\pi \ta}+e^{i\tnu/2}\frac{\sin\pi(2\ta+m)}{\sin 2\pi \ta}.
\end{equation}
Similarly, from \eqref{eq:MBA}, \eqref{eq:MBB} one obtains the equality
\begin{equation}\label{rel:MBs}
\tr M_B=e^{-i\nu/2}\frac{\sin\pi(2a-m)}{\sin 2\pi a}+e^{i\nu/2}\frac{\sin\pi(2a+m)}{\sin 2\pi a}=e^{2\pi i \ta}+e^{-2\pi i \ta}.
\end{equation}
Using the same expressions as above, we further obtain that
\begin{align}
\tr M_AM_B&=e^{-i\nu/2+2\pi ia}\frac{\sin\pi(2a-m)}{\sin 2\pi a}+e^{i\nu/2-2\pi ia}\frac{\sin\pi(2a+m)}{\sin 2\pi a}=\nonumber
\\
&=e^{-i\tnu/2-2\pi i\ta}\frac{\sin\pi(2\ta-m)}{\sin2\pi \ta}+e^{i\tnu/2+2\pi i\ta}\frac{\sin\pi(2\ta+m)}{\sin 2\pi \ta}.\label{rel:MABs}
\end{align}

Let us now analyse the expressions above. From \eqref{rel:MBs} one can get the expression for $a$ in terms of \(\ta\), $\nu$:
\begin{equation}
(e^{4\pi ia}-1)(e^{2\pi i\ta}+e^{-2\pi i \ta})=e^{-i\nu/2}(e^{4\pi i a-\pi i m}-e^{\pi i m})+e^{i\nu/2}(e^{4\pi i a+\pi i m}-e^{-\pi i m}),
\end{equation}
and therefore we get \eqref{rel:a_as_atnu}
\begin{equation}\label{eq:aeqatmnu}
e^{4\pi i a}=\frac{e^{2\pi i\ta}+e^{-2\pi i \ta}-e^{\pi im-i\nu/2}-e^{-\pi i m+ i \nu/2}}{e^{2\pi i\ta}+e^{-2\pi i\ta}-e^{\pi im+i\nu/2}-e^{-\pi i m-i \nu/2}}.
\end{equation}
The variable $\tnu$ can be expressed in terms of $a$, $\ta$, $\nu$ by considering the following combination:
\begin{align}
&\tr M_A-e^{2\pi i \ta}\tr M_AM_B \nonumber\\
&=
e^{2\pi i a}+e^{-2\pi ia}-e^{2\pi i \ta}\left(e^{-i\nu/2+2\pi ia}\frac{\sin\pi(2a-m)}{\sin 2\pi a}+e^{i\nu/2-2\pi ia}\frac{\sin\pi(2a+m)}{\sin 2\pi a}\right)
\nonumber\\
&=
(1-e^{4\pi i\ta})e^{i\tnu/2}\frac{\sin\pi(2\ta+m)}{\sin 2\pi \ta}.
\end{align}
Comparing the LHS and RHS of the equation above, we get:
\begin{align}
e^{i\tnu/2}&=\frac{\sin 4\pi a-e^{2\pi i \ta}\left( e^{-i\nu/2+2\pi ia}\sin\pi(2a-m) + e^{i\nu/2-2\pi ia}\sin\pi(2a+m) \right)}{-2i \sin2\pi a\sin\pi(2\ta+m)}
\nonumber \\
&=e^{2\pi ia}\frac{\sin 4\pi a-e^{2\pi i \ta}\left( e^{-i\nu/2+2\pi ia}\sin\pi(2a-m) + e^{i\nu/2-2\pi ia}\sin\pi(2a+m) \right)}{(1-e^{4\pi ia}) \sin2\pi a\sin\pi(2\ta+m)}.
\end{align}
Substituting \eqref{rel:a_as_atnu} in the above expression and simplifying it, we get \eqref{rel:nut_as_aatnu}.
\end{proof}

One can obtain an analogous result to Proposition \ref{prop:SMonodromy}  for the pair of variables $\ta$, $\nu$:
\begin{equation}\label{eq:nueqamnut}
\begin{split}
\ta&=\frac1{4\pi i}\log\frac{\sin(\pi a-\pi m/2+\tnu/4)\sin(\pi a+\pi m/2-\tnu/4)}{\sin(\pi a-\pi m/2-\tnu/4)\sin(\pi a+\pi m/2+\tnu/4)},\\
\nu&=4\pi \ta-2i\log \frac{\sin(\pi a-\pi m/2-\tnu/4)}{\sin(\pi a+\pi m/2-\tnu/4)}.
\end{split}
\end{equation}

We can obtain further constraint on the monodromy data as follows: since there is overall \(SL(2)\) action on the \(C\) matrices, without loss of generality, one can choose \(C_0^{A}=C^{B}_0=1\). This in turn implies that
\begin{equation}
X_0\mathop{=}^{\eqref{eq:C0A}}
\begin{pmatrix}
e^{-i\delta_0} & 0\\
0 & e^{i\delta_0}
\end{pmatrix}^{-1}
\begin{pmatrix}
\frac{e^{\pi i a}\Gamma(1-2a)\Gamma(-2m)}{\Gamma(1-2a-m)\Gamma(-m)} & -\frac{e^{\pi i a}\Gamma(1-2a)\Gamma(2m)}{\Gamma(m)\Gamma(1-2a+m)}\\
-\frac{e^{-\pi i a}\Gamma(2a)\Gamma(-2m)}{\Gamma(2a-m)\Gamma(-m)} & -\frac{e^{-\pi ia} \Gamma(2a) \Gamma(2m)}{\Gamma(m)\Gamma(2a+m)}
\end{pmatrix}^{-1}
\end{equation}
and
\begin{equation}
X_{\infty}\mathop{=}^{\eqref{eq:C0B}}
\begin{pmatrix}
e^{-i\widetilde{\delta}_0} & 0\\
0 & e^{i\widetilde{\delta}_0}
\end{pmatrix}^{-1}
\begin{pmatrix}
\frac{e^{-\pi i \ta}\Gamma(1-2\ta)\Gamma(-2m)}{\Gamma(1-2\ta-m)\Gamma(-m)} & -\frac{e^{-\pi i\ta}\Gamma(1-2\ta)\Gamma(2m)}{\Gamma(m)\Gamma(1-2\ta+m)}\\
-\frac{e^{\pi i \ta}\Gamma(2\ta)\Gamma(-2m)}{\Gamma(2\ta-m)\Gamma(-m)} & -\frac{e^{\pi i \ta}\Gamma(2\ta)\Gamma(2m)}{\Gamma(m)\Gamma(2\ta+m)}
\end{pmatrix}^{-1}.
\end{equation}

\begin{remark}
    The coordinates $\tr M_A$, $\tr M_B$ and $\tr M_AM_B$ satisfy the character variety relation (see \cite{NRS2011} for example).
    \begin{equation}\label{eq:Fricke}
        \tr M_0=(\tr M_B)^2-(\tr M_A)(\tr M_B)(\tr M_AM_B)+(\tr M_AM_B)^2+(\tr M_A)^2-2
    \end{equation}
    which is known as the Fricke cubic. Performing the transformation twice does not give the original coordinates, i.e. $\left(\widetilde{\at},\widetilde{\nut}\right) \ne\left(a,\nu\right)$. This is because, while \eqref{rel:MAs} and \eqref{rel:MBs} are invariant under this operation, \eqref{rel:MABs} is not. Instead, doing the transformation twice exchanges the two roots of \eqref{eq:Fricke} seen as a quadratic equation for $\tr M_AM_B$, namely
    \begin{equation}
        \tr M_AM_B\mapsto\tr M_A\tr M_B-\tr M_AM_B.
    \end{equation}

We can notice that
\begin{align}
\tr M_A\tr M_B-\tr M_AM_B&=e^{-i\nu/2-2\pi ia}\frac{\sin\pi(2a-m)}{\sin 2\pi a}+e^{i\nu/2+2\pi ia}\frac{\sin\pi(2a+m)}{\sin 2\pi a}=\nonumber
\\
&=e^{-i\tnu/2+2\pi i\ta}\frac{\sin\pi(2\ta-m)}{\sin2\pi \ta}+e^{i\tnu/2-2\pi i\ta}\frac{\sin\pi(2\ta+m)}{\sin 2\pi \ta}.
\end{align}
Therefore, if we fix the values of \(a\) and \(\widetilde{a}\), this involution acts by
\begin{equation}
\nu\mapsto -\nu,\qquad \widetilde{\nu}\mapsto -\widetilde{\nu}, \qquad m\mapsto -m.
\end{equation}
\end{remark}

We can further choose the following expressions for $\delta_0$, $\widetilde{\delta}_0$ so as to simplify the expressions of $M_A$, $M_B$:
\begin{equation}
\delta_0=\Delta_0-\frac{i}{2} \log \frac{2^{-4m}\Gamma(2a+m)\Gamma(1/2-m)}{\Gamma(2a-m)\Gamma(1/2+m)},
\end{equation}
and
\begin{equation}
\widetilde{\delta}_0=\widetilde{\Delta}_0-\frac{i}{2}\log \frac{2^{-4m}\Gamma(1-2\ta+m)\Gamma(1/2-m)}{\Gamma(1-2\ta-m)\Gamma(1/2+m)}.
\end{equation}
Finally, comparing the off-diagonal elements of monodromy matrices \eqref{eq:MAB} and \eqref{eq:MBA}, we get
\begin{equation}\label{eq:diffdelta0s}
\Delta_0-\widetilde{\Delta}_0=\frac{i}{2}\log \frac{-i e^{-\pi i m+i \nu/2}\sin\pi(2\ta-m)}{\cos \pi m-e^{i\nu/2}\cos 2\pi \ta}.
\end{equation}

\begin{remark}[Redundancies in monodromy coordinates]
The variables $(a,\,\nu,\,m)$ are not actual coordinates on the character variety, since certain discrete transformations are mapped to the same point in monodromy space. We list them here:
\begin{enumerate}
    \item The transformation \(m\mapsto-m\) is a symmetry of the original problem, but to preserve monodromy data we should also transform \(\nu\) and \(\widetilde{\nu}\):
\begin{equation}
\label{eq:msymmetry}
m\mapsto-m,\qquad
\nu\mapsto \nu+2i\log\frac{\sin\pi(2a+m)}{\sin\pi(2a-m)},\qquad
\widetilde{\nu}\mapsto \widetilde{\nu}+2i\log\frac{\sin\pi(2\widetilde{a}+m)}{\sin\pi(2\widetilde{a}-m)};
\end{equation}
    \item Coordinates of opposite sign describe the same point in the character variety:
\begin{equation}
\label{eq:invsymmetry1}
\nu\mapsto-\nu,\qquad a\mapsto-a
\end{equation}
and
\begin{equation}
\label{eq:invsymmetry2}
\widetilde{\nu}\mapsto-\widetilde{\nu},\qquad \widetilde{a}\mapsto-\widetilde{a};
\end{equation}
    \item We also have the four trivial transformations that can be applied independently, signaling the fact that the monodromies only depend on exponentiated variables:
\begin{equation}
\label{eq:shiftsymmetries}
a\mapsto a+k_a, \quad
\widetilde{a}\mapsto \widetilde{a}+k_{\widetilde{a}},\quad
\nu\mapsto \nu+4\pi  k_{\nu},\quad
\widetilde{\nu}\mapsto \widetilde{\nu}+4\pi k_{\widetilde{\nu}},\quad
m\mapsto m+k_m,\quad
k_{\bullet}\in\mathbb{Z}.
\end{equation}
\end{enumerate}
Summarizing this we notice that monodromy manifold is not \(\mathbb{C}^3\) with coordinates \(m,a,\,\nu\) or \(m,\,\at,\,\nut\), but instead
\begin{equation}
\mathbb{C}^3/\{\eqref{eq:msymmetry},\eqref{eq:invsymmetry1},\eqref{eq:shiftsymmetries}\}\simeq
\mathbb{C}^3/\{\eqref{eq:msymmetry},\eqref{eq:invsymmetry2},\eqref{eq:shiftsymmetries}\}.
\end{equation}
The tau function is a function on $\mathbb{C}^3_{(a,\nu,m)}$, but it doesn't descend to a function on the quotient. Rather, it is a section of a line bundle on the quotient.
\end{remark}

\section{The modular connection constant}\label{sec:Upsilon}

The data obtained in the previous section sets the stage to compute the connection constant.
\begin{theorem}\label{eq:thm:ConnCon}
    The connection constant \eqref{def:Upsilon} is given by
    \begin{align}
        \Upsilon_S &=  \frac{G(1+2a)G(1-2a)}{G(1-m+2a)G(1-m-2a)}\frac{G(1-m+2\ta)G(1-m-2\ta)}{G(1+2\ta)G(1-2\ta) }\,\widehat{\Upsilon}_S(a,\nut), \label{eq:ConnConst}
    \end{align}
    where
    \begin{equation}\label{eq:Uhat}
        \widehat{\Upsilon}_S(a,\nut)=e^{i\,\ta\, \tnu}\frac{\widehat{G}\left(a-m/2+\tnu/(4\pi)\right)\widehat{G}\left(a-m/2-\tnu/(4\pi)\right)}{\widehat{G}\left(a+m/2+\tnu/(4\pi)\right)\widehat{G}\left(a+m/2-\tnu/(4\pi)\right)}\frac{(2\pi)^m\widehat{G}(m)}{e^{i\pi m^2/2}},
    \end{equation}
    \begin{equation}\label{def:Ghat}
\widehat{G}(x)=\frac{G(1+x)}{G(1-x)},
\end{equation}
and $G(.)$ denotes the Barnes G-function.
\end{theorem}

\begin{proof}
Let us begin by computing the following term
\begin{align}
- \omega^{A}_{3pt}+ {\omega}^B_{3pt}&\mathop{=}^{\eqref{def:omega3A},\eqref{def:omega3B}}\tr \mathbf{a}\,dG_{-}^{A}\left(G_{-}^{A}\right)^{-1} - \tr \mathbf{a} \, dG_{+}^{A}\left(G_{+}^{A}\right)^{-1}- \tr \mathbf{m} \, dG_0^{A}\left(G_0^{A}\right)^{-1} \nonumber\\
&-\left( \tr \widetilde{\mathbf{a}}\, dG_{-}^{B}\left(G_{-}^{B}\right)^{-1} + \tr \widetilde{\mathbf{a}}\, dG_{+}^{B}\left(G_{+}^{B}\right)^{-1}+\tr \mathbf{m}\, dG_0^{B}\left(G_0^{B}\right)^{-1}\right)\nonumber\\
&\mathop{=}^{\textrm{Props.} \ref{Prop:GAs}, \ref{Prop:GBs}}-a d \left( i\nu + 2 \log \frac{\Gamma(1-2a)\Gamma(2a-m)}{\Gamma(2a)\Gamma(1-2a-m)} \right)-m d\log \frac{2^{-4m}\Gamma(2a+m)\Gamma(1/2-m)}{\Gamma(2a-m)\Gamma(1/2+m)}
\nonumber\\
&+\ta d \left( i\tnu+ 2\log \frac{\Gamma(1-2\ta)\Gamma(2\ta-m)}{\Gamma(1-2\ta-m)} \right)+m d\log \frac{2^{-4m}\Gamma(1-2\ta+m)\Gamma(1/2-m)}{\Gamma(1-2\ta-m)\Gamma(1/2+m)}\nonumber\\
&-2im\, d(\Delta_0-\widetilde{\Delta}_0) \nonumber\\
&\mathop{=}^{\eqref{eq:diffdelta0s}}-d\log \frac{G(1-m+2a)G(1-m-2a)}{G(1+2a)G(1-2a)}\frac{G(1+2\ta)G(1-2\ta)}{G(1-m+2\ta)G(1-m-2\ta)}\left( \frac{\sin\pi(2\ta-m)}{\sin\pi(2a+m)} \right)^m
\nonumber\\
&-dm \,\log \frac{\sin\pi(2a+m)}{\sin\pi(2\ta-m)}+m d\log \frac{-i e^{-\pi i m+i\nu/2}\sin\pi(2\ta-m)}{\cos\pi m-e^{i\nu/2}\cos\pi \ta}-i ad\nu+i\ta d\tnu
\nonumber\\
&=d\log \frac{G(1+2a)G(1-2a)}{G(1-m+2a)G(1-m-2a)}\frac{G(1-m+2\ta)G(1-m-2\ta)}{G(1+2\ta)G(1-2\ta) }\nonumber\\
&+m d\log \frac{-i e^{-\pi i m+i\nu/2}\sin\pi(2 a+m)}{\cos\pi m-e^{i\nu/2}\cos\pi \ta}-i ad\nu+ i\ta d\tnu.\label{comp:conn-const}
\end{align}
For the further analysis we define the new 1-form
\begin{equation}
\omega_1=-i\,a\,d\nu+i\,\ta\, d\tnu+m\,d\log \frac{-i e^{-\pi i m+i\nu/2}}{\cos\pi m-e^{i\nu/2}\cos\pi \ta}.
\end{equation}
We can check by explicit computation that this 1-form is closed: rewriting the above expression as
\begin{equation}
\omega_1=-i \, a\,d\nu + i\, \ta\, d \tnu+2\pi i \,m\, d \ta+\pi i\, m\,dm+m\, d\log \frac{\sin(\pi a -\pi m/2-\tnu/4)\sin(\pi a +\pi m/2+\tnu/4)}{e^{\pi i m-2\pi i\ta}\sin\pi m},
\end{equation}
one can check that this 1-form is an external logarithmic derivative of the function \(\widehat{\Upsilon}_S\):
\begin{align}
\omega_1=d\log \widehat{\Upsilon}_S, &&
\widehat{\Upsilon}_S=e^{i\,\ta\, \tnu}\frac{(2\pi)^m\widehat{G}(m)}{e^{i\pi m^2/2}}\frac{\widehat{G}\left(a-m/2+\tnu/(4\pi)\right)\widehat{G}\left(a-m/2-\tnu/(4\pi)\right)}{\widehat{G}\left(a+m/2+\tnu/(4\pi)\right)\widehat{G}\left(a+m/2-\tnu/(4\pi)\right)},\label{def:Upsilon1}
\end{align}
with $\widehat{G}$ defined in \eqref{def:Ghat}. Then substituting \eqref{def:Upsilon1} in \eqref{comp:conn-const}, we obtain the statement of the theorem.
\end{proof}

\section{Modular transformation of the tau function}\label{sec:SDual}
In this section, we detail the action of the modular transformation $\tau\mapsto-1/\tau$ on the tau function, justifying the name of ``modular'' connection constant from the previous section. We first note that the modular group $SL(2,\mathbb{Z})$ of the torus is a symmetry of the isomonodromic system, as can be easily verified. 
\begin{proposition}\label{prop:ModNAECM}
Let $\taut:=-1/\tau$.  The following properties hold:
\begin{align}\label{eq:modtraPQ}
\Qt:=Q(\taut; \ta,\nut) = -Q(\tau; a,\nu)/\tau, && \Pt:=P(\taut; \ta, \nut) = -\tau P(\tau; a, \nu)+ 2\pi i Q(\tau; a,\nu).
\end{align}
Furthermore, $\Qt$ solves the equation 
\begin{align}\label{modeq}
(2\pi i \partial_{\widetilde{\tau}})^2 \widetilde{Q} = m^2 \wp'\left(2 \widetilde{Q} \vert \widetilde{\tau} \right).
\end{align}
\end{proposition}
\begin{proof}
The derivative of Weierstrass $\wp$-function has the following transformation (see \cite{manin1996sixth} (1.12)):
\begin{align*}
    \wp'\left(\frac{z}{c_3\tau+c_4}\Big| \frac{c_1\tau+c_2}{c_3\tau+c_4} \right) = (c_3\tau+c_4)^3 \wp'\left(z\vert \tau \right),
\end{align*}
Therefore, if $Q(\tau)$ and $P(\tau)$ solve equation \eqref{eq:NAECM}, then the modified variables
\begin{align}
(c_3\tau+c_4)Q\left(\frac{c_1\tau+c_2}{c_3\tau+c_4} \right), && \frac{P\left(\frac{c_1\tau+c_2}{c_3\tau+c_4}\right)}{c_3\tau+c_4}+2\pi ic_3\,Q\left(\frac{c_1\tau+c_2}{c_3\tau+c_4}\right)
\end{align}
solve the equation \eqref{modeq}.
\end{proof}
We can study the modular transformation of the linear system in a similar manner. Indeed, we obtain the following result.
\begin{corollary}
The Lax matrix $L_{z}$ and the solution of the linear system transform as
\begin{align}\label{def:Ltil}
    \widetilde{L}_{z} = e^{2\pi i Q(\tau) z \sigma_3/\tau} L_{z} e^{-2\pi i Q(\tau) z \sigma_3/\tau} - 2\pi i Q \sigma_3/\tau, &&    \widetilde{Y} = C_S Y e^{-2\pi i Q(\tau) z \sigma_3/\tau}, 
\end{align}
 where $C_S$ determines the choice of normalisation.
\end{corollary}
\begin{proof}
    Let us begin with the linear system under the transformation $z\to z\, \widetilde{\tau}$:
    \begin{align}\label{eq:transf_linsys}
      \frac{1}{\widetilde{\tau}} Y^{-1}\partial_z Y = \left(\begin{array}{cc}
        \widetilde{P}(\tau) & m\,x(-2\widetilde{Q}(\tau),z\, \widetilde{\tau}) \\
        m\,x(2\widetilde{Q}(\tau),z\, \widetilde{\tau}) & -\widetilde{P}(\tau)
    \end{array}\right).
    \end{align}
    To get the above expression we look at the modular transformation of $x(\xi, z)$. We know that
    \begin{align}
        (- i \tau)^{1/2} \theta_1(z\vert \tau) = -i e^{i\pi z^2 \ttau} \theta_1(z\ttau\vert \ttau), &&  (- i \tau)^{1/2} \theta_1'(0\vert \tau) = - i \ttau \theta_1'(0\vert \ttau), \label{mod:theta}
    \end{align}
    and so,
    \begin{align}
         x(\xi\, \widetilde{\tau}, z\, \widetilde{\tau}) = \frac{\theta_{1}((z- \xi)\widetilde{\tau}) \theta_{1}'(0\vert \widetilde{\tau})}{\theta_{1}(z\, \widetilde{\tau}) \theta_{1}(\xi\, \widetilde{\tau})} = -\tau e^{-2\pi i z \xi /\tau} x(\xi, z). 
    \end{align}
    The RHS of \eqref{eq:transf_linsys}, with the expressions \eqref{eq:modtraPQ} reads
    \begin{align}
        \frac{1}{\widetilde{\tau}} Y^{-1}\partial_z Y &= \left(\begin{array}{cc}
        \widetilde{P}(\tau) & m\,x(-2\widetilde{Q}(\tau),z\, \widetilde{\tau}) \\
        m\,x(2\widetilde{Q}(\tau),z\, \widetilde{\tau}) & -\widetilde{P}(\tau)
    \end{array}\right) \nonumber\\
    &= \left(\begin{array}{cc}
        - \tau P + 2\pi i Q & -\tau e^{4\pi i Q(\tau)z/\tau} m\,x(-2{Q}(\tau),z) \\
      -\tau e^{-4\pi i Q(\tau) z/\tau} m\,x(2{Q}(\tau),z) &  \tau P - 2\pi i Q
    \end{array}\right) \nonumber\\
    & = -\tau \left(e^{2\pi i Q(\tau) z \sigma_3/\tau} L_{CM} e^{-2\pi i Q(\tau) z \sigma_3/\tau} - 2\pi i Q \sigma_3/\tau \right).
    \end{align}
    Therefore, the solution transforms as
    \begin{align}
        \widetilde{Y} = C_S Y e^{-2\pi i Q(\tau) z \sigma_3/\tau}. 
    \end{align}

\end{proof}
In the following theorem, we will relate precisely the normalised tau function $\T^B$ to the modular transformation of the tau function $\T^A$.
\begin{theorem}\label{thm:taumodular}
	The modular transformation $\widetilde{\T}^A$ of the tau function $\T^A$, depending on monodromy coordinates $\at,\,\nut,\,m$ is given by
	\begin{equation}\label{eq:SDualTau}
		\widetilde{\T}^A(\ttau;\ta,\nut,m)=\frac{1}{\Upsilon_S(a,\nu,m)}\,\T^A(\tau;a,\nu,m)\,e^{-2\pi i Q(\tau;a,\nu)^2/\tau}\,\tau^{m^2}  e^{- i\pi m^2}.
	\end{equation}
 where $\Upsilon_S$ is the modular connection constant \eqref{eq:ConnConst}. Here $\widetilde{\T}^A$ is given by
     \begin{equation}
    	\dd\log\widetilde{\T}^A:=2\Pt\dd_\cM \Qt+\Ht\frac{\dd\tau}{2\pi i}+\tr\left(m\sigma_3\dd_\cM \Gt\,\Gt^{-1} \right)-\omega_{3pt}^B,
    \end{equation}
    where the dual Hamiltonian is
\begin{equation}\label{def:dualHam}
	\widetilde{H}:=\widetilde{P}^2-m^2\wp(2\widetilde{Q}|\ttau)-2m^2\eta_1(\ttau),
\end{equation}
and the matrix $\widetilde{G}$ diagonalises $\widetilde{L}_z$ in \eqref{def:Ltil} near $z=0$
\begin{align}
    \res_{z=0}\widetilde{L}_{z}=-m\sigma_1=:\widetilde{G}^{-1}m\sigma_3 \widetilde{G}.
\end{align}
\end{theorem}
\begin{proof}
    To prove the theorem, we consider the logarithmic exterior derivatives of the tau function. Recall that
    \begin{equation}
    	\dd\log\T^A=2P\dd_\cM Q+H\frac{\dd\tau}{2\pi i}+\tr\left(m\sigma_3\dd_\cM G\,G^{-1} \right)-\omega_{3pt}^A.
    \end{equation}

Let us now study the form of the dual Hamiltonian \eqref{def:dualHam}. From equation \eqref{eq:modtraPQ}, we have
\begin{equation}
	2\Pt\dd_\cM \Qt=2\left(-\tau P+2\pi i Q \right)\dd_\cM\left(-Q/\tau \right)=2P\dd_\cM Q-\frac{4\pi i}{\tau}Q\dd_\cM Q.
\end{equation}
Moreover, the following identities of theta function for $\ttau = -1/\tau$ \cite[20.7]{NIST:DLMF}:
\begin{align*}
   (- i \tau)^{1/2} \theta_1(z\vert \tau) = -i e^{i\pi z^2 \ttau} \theta_1(z\ttau\vert \ttau), &&  (- i \tau)^{1/2} \theta_1'(0\vert \tau) = - i \ttau \theta_1'(0\vert \ttau),
\end{align*}
imply that
\begin{equation}\label{mod:eta}
	\eta_1(\ttau)=\tau^2\eta_1-i\pi\tau.
\end{equation}
Substituting the above expressions in \eqref{def:dualHam}, we have
\begin{equation}\label{eq:TrasfHam}
\begin{split}
	\Ht & = \Pt^2-m^2\wp(2\Qt|\taut)-2m^2\teta_1=(2\pi i Q-\tau P)^2-\tau^2 m^2\wp(2Q|\tau)-2m^2(\tau^2\eta_1-i\pi\tau) \\
	& =\tau^2 H-4\pi i\tau P Q-4\pi^2 Q^2+2\pi i m^2\tau.
\end{split}
\end{equation}

The remaining step to prove this theorem is to find the relation between $\Gt$ and $G$. To this end, consider the modular transformation of solution $Y$ given by \eqref{def:Ltil}: 
\begin{equation}
	\Yt(z)=Y(-z/\tau,\ttau;\at,\nut,m)=Y(z,\tau;a,\nu,m)e^{-2\pi i zQ\sigma_3/\tau}.
\end{equation}
Its local behaviour around the origin will then be
\begin{equation}
	\Yt(z)\simeq C_0\left(-\frac{z}{\tau} \right)^{m\sigma_3}\Gt=C_0z^{m\sigma_3}G e^{-2\pi izQ \sigma_3/\tau}.
\end{equation}
Using remark \ref{rmk:minus1}, we see that the branch for $(-1)^m$ has to be chosen so that $(-z/\tau)^{m\sigma_3}=e^{i\pi m\sigma_3}(z/\tau)^{m\sigma_3} $, implying that
\begin{equation}
\Gt= e^{-i \pi m\sigma_3}\tau^{m\sigma_3}G.
\end{equation}
Therefore,
\begin{equation}\label{transfTerm2}
\tr\left(m\sigma_3\dd_{\cM}\Gt\,\Gt^{-1} \right)=\tr\left(m\sigma_3 \dd_{\cM} G\,G^{-1}\right)+2m\log\tau\,\dd m - 2\pi i m \dd m.
\end{equation}
Putting together the equations \eqref{eq:TrasfHam}, \eqref{transfTerm2}, we find
\begin{equation}
\begin{split}
	\dd\log\Tt^A & = 2P\dd_{\cM}Q-\frac{4\pi i}{\tau}Q\dd_{\cM}Q+\left(\tau^2 H-4\pi i\tau P Q-4\pi^2 Q^2+2\pi i m^2\tau \right)\frac{\dd\tau}{2\pi i \tau^2} \\
	& +2m\log\tau\,\dd m- 2\pi i m \dd m-\omega_{3pt}^B \\
	& =\dd\log\T^A+\omega_{3pt}^A-\omega_{3pt}^B-2\pi i\dd_{\cM}\left(Q^2/\tau \right)-2\pi i\dd_\tau\left(Q^2/\tau\right)+\dd_\tau\log\tau^{m^2}\\
    &+\dd_{\cM}\log\tau^{m^2} + \dd_{\mathcal{M}} e^{- i\pi m^2} \\
	& =\dd\log\left(\Upsilon_S^{-1}\T^A e^{-2\pi i Q^2/\tau}\tau^{m^2}  e^{- i\pi m^2}\right).
\end{split}
\end{equation}
Integrating the above expression, we find \eqref{eq:SDualTau}, proving this theorem.
\end{proof}
Through Theorem \ref{thm:taumodular}, we obtain the following characterisation of the connection constant $\Upsilon_S$, as a ratio of different asymptotics of isomonodromic tau functions:
\begin{corollary}
    The modular connection constant $\Upsilon_S$ can be characterised in the following manner:
    \begin{equation}
        e^{i\pi m^2} \Upsilon_S=e^{i\at(\nut+\delta\nut)}\frac{\lim_{\tau\rightarrow 0}\left( \tau^{m^2}\T(\tau)\right)}{\lim_{\tau\rightarrow i\infty}\left(e^{-2\pi ia^2\tau}\T(\tau)\right)}.
    \end{equation}
\end{corollary}
\begin{proof}
 To prove this relation, let us consider first the $\tau\rightarrow i\infty$ behaviour of the tau function $\T^A$. This can be computed through the relation from 
 \cite[eq. (3.51)]{Bonelli2019}, proved in \cite{DMDG2020}:
\begin{equation}\label{eq:this1}
	\frac{1}{\eta(\tau)^2}\T(a,\nu,m,\tau)\theta_1(Q(\tau)+\rho|\tau)\theta_1(Q(\tau)-\rho|\tau)=Z_D(a,\nu,m,\tau,\rho),
\end{equation}
\begin{equation}\label{eq:ZDCB}
	Z_D(a,\nu,m,\tau,\rho)=\frac{1}{\eta(\tau)}\sum_{n,k \in \mathbb{Z}}e^{\frac{i n}{2}(\nu+\delta\nu)}e^{2\pi i\tau\left(k+\frac{n}{2}+\frac{1}{2} \right)^2}e^{4\pi i \left(k+\frac{n}{2}+\frac{1}{2} \right)\left(\rho+\frac{1}{2} \right)}\cB\left(a+\frac{n}{2},m,\tau \right).
\end{equation}
 As $\tau\rightarrow i\infty$, we have
 \begin{equation}
     \eta(\tau)\sim e^{\frac{i\pi \tau}{12}},\qquad Q\sim a\tau+\frac{\nu+\delta\nu}{4\pi} ,\qquad\theta_1(Q|\tau)\sim -ie^{\frac{i\pi\tau}{4}}e^{i\pi (Q\pm\rho)}. 
 \end{equation}
The function $\cB$ has the form
\begin{equation}
    \cB(a,m,\tau)=e^{2\pi i\tau a^2}\sum_{\ell=0}^{\infty}e^{2\pi i \ell\tau}\cB_\ell(a,m),\qquad \cB_0=1.
\end{equation}
Then the asymptotics of $Z_D$ is dominated by the term $k=0,\,n=-1$, and
\begin{equation}
    Z_D(a,\nu,m,\tau,\rho)\sim e^{-\frac{i\pi\tau}{12}}e^{-\frac{i}{2}(\nu+\delta\nu)}e^{2\pi i\tau a^2}e^{-2\pi i\tau a}.
\end{equation}
 Then,
 \begin{equation}
     \T^A\mathop{\sim}_{\tau\rightarrow i\infty}
     e^{2\pi i\tau a^2}.
 \end{equation}
 On the other hand, from \eqref{eq:SDualTau} we have
 \begin{gather}
     \T^A=\Upsilon_S e^{2\pi i Q^2/\tau}\tau^{-m^2}\Tt\mathop{\sim}_{\tau\rightarrow 0}
     =e^{2\pi i\at^2\taut}e^{2\pi i\left(\at-\frac{\nut+\delta\nut}{4\pi} \right)^2/\tau^{1-m^2}} \nonumber\\
     =\Upsilon_S 
     e^{-i\at(\nut+\delta\nut)}\tau^{-m^2}  e^{i\pi m^2}.
 \end{gather}
 The ratio of the two equations gives the proof.
\end{proof}

\section{The $c=1$ Virasoro modular kernel}\label{sec:c1}

In this section, we will use equation \eqref{eq:Uhat} for the connection constant to compute the $c=1$ modular kernel of Virasoro conformal blocks. 

\begin{lemma}\label{thm:UShift}
    The modular connection constant $\widehat{\Upsilon}_S$ in \eqref{eq:ConnConst} satisfies
    \begin{equation}\label{eq:UShift}
        \widehat{\Upsilon}_S\left(\at+\frac{n}{2},a \right)=e^{\frac{i\nu n}{2}}\widehat{\Upsilon}_S(\at,a).
    \end{equation}
\end{lemma}
\begin{proof}
    As is written in \eqref{eq:ConnConst}, $\widehat{\Upsilon}_S$ depends explicitly on $a,\at,\nut$ and $m$. Under the shift $\at\mapsto\at+\frac{n}{2}$ $m$ is left unchanged, so we have to understand what happens to $\nut=\nut(a,\at)$. In fact, we have to first understand the behaviour of $\nu(a,\at)$. Consider the asymptotics as $\nu\rightarrow+i\infty$ of \eqref{eq:aeqatmnu}:
    \begin{equation}
        e^{4\pi ia}\mathop{\approx}_{\nu\rightarrow+i\infty}\frac{e^{i\pi m-\frac{i\nu}{2}}-e^{-2\pi i\at}}{e^{-i\pi m-\frac{i\nu}{2}}-e^{-2\pi i\at}}.
    \end{equation}
    From this equation we see that
    \begin{equation}
        \nu\left(\at+\frac{1}{2},a\right)=\nu(\at,a)+2\pi.
    \end{equation}
    To understand what happens to $\nut$, consider equation \eqref{eq:nueqamnut}, that we reproduce here for convenience
    \begin{equation}
\nu=4\pi \ta-2i\log \frac{\sin(\pi a-\pi m/2-\tnu/4)}{\sin(\pi a+\pi m/2-\tnu/4)}.
\end{equation}
From the transformation of $\nu$ derived above, we see that the second term in the equation above must be invariant. This means that $\nut$ must be invariant. This is because in order for the ratio of sine functions to be unchanged, one needs $\nut$ to be shifted by integer multiples of $4\pi$. However, this would amount to a closed loop around either a zero of the pole of the argument of the logarithm, so that the logarithm would undergo monodromy from crossing of a log-cut. Then \eqref{eq:UShift} follows from the exponential factor $e^{i\at\nut}$ in \eqref{eq:ConnConst}.
\end{proof}

Recall the relation between the tau function $\T$ and $c=1$ Virasoro conformal blocks $\cB$ \eqref{eq:this1}. The same relation can be derived in an identical manner for the isomonodromic problem and conformal blocks in the S-dual channel $\tau\rightarrow-1/\tau$:
\begin{equation}
	\frac{1}{\eta(\taut)^2}\T(\at,\teta,m,\taut)\theta_1(\Qt(\taut)+\rhot|\taut)\theta_1(\Qt(\taut)-\rhot|\taut)=Z_D(\at,\nut,m,\taut,\rhot),
\end{equation}
\begin{equation}
	Z_D(\at,\nut,m,\taut,\rhot)=\frac{1}{\eta(\taut)}\sum_{n,k}e^{\frac{i\nut n}{2}}e^{2\pi i\taut\left(k+\frac{n}{2}+\frac{1}{2} \right)^2}e^{4\pi i \left(k+\frac{n}{2}+\frac{1}{2} \right)\left(\rhot+\frac{1}{2} \right)}\cB\left(\at+\frac{n}{2},m,\taut \right).
\end{equation}
The conformal blocks can then be obtained by simply inverting the Fourier series:
\begin{align}
	\cB(a,m,\tau)&=\int_0^1\dd\rho\int_0^{4\pi}\frac{\dd\nu}{4\pi}\,\eta(\tau)e^{-i \pi \tau/2} e^{-2\pi i (\rho +1/2)}Z_D(a,\nu,m,\tau,\rho), \label{eq:this2}\\
 \cB(\at,m,\taut)&=\int_0^1\dd\rhot\int_0^{4\pi}\frac{\dd\nut}{4\pi}\, \eta(\taut) e^{-i \pi \taut/2} e^{-2\pi i (\rhot +1/2)} Z_D(\at,\nut,m,\taut,\rhot)
\end{align}

We will now use the connection constant \eqref{def:Upsilon1}, together with the modular transformation of the tau function \eqref{eq:SDualTau}, to obtain the $c=1$ Virasoro modular kernel $S(a,\at)$, defined by
\begin{equation}\label{eq:Skerneldef}
	\cB(a,m,\tau)=(e^{-i\pi}\tau)^{-\Delta(m)}\int_{-\infty+i\Lambda}^{\infty+i\Lambda}\dd\at\, S(a,\at)\cB(\at,m,\taut).
\end{equation}
\begin{theorem}\label{thm:c1ker}
    The $c=1$ Virasoro modular kernel is given by
    \begin{equation}\label{eq:c1modular}
        S(a,\at)=\frac{\sqrt{2}}{4\pi}\frac{\partial\nu}{\partial\at}\,\widehat{\Upsilon}_S(a,\at),
    \end{equation}
    where $\widehat{\Upsilon}_S(a,\at)$ is the modular connection constant \eqref{eq:ConnConst}.
\end{theorem}
\begin{proof}
First, let us use the modular transformation of $\theta_1$ and of the Dedekind eta function $\eta(\taut)=\sqrt{-i\tau}\,\eta(\tau)$ in \eqref{mod:theta}, \eqref{mod:eta} to write the relation between dual partition functions:
\begin{gather}
	\Zt_D=\frac{\theta_1(\Qt+\rhot|\taut)\theta_1(\Qt-\rhot|\taut)}{\eta(\taut)^2}\Tt=-e^{2\pi i(Q^2+\rho^2)/\tau}\frac{\theta_1(Q+\rho|\tau)\theta_1(Q-\rho|\tau)}{\eta(\tau)^2}\Tt \nonumber \\
	\mathop{=}^{\eqref{eq:SDualTau}} -\frac{1}{\widehat{\Upsilon}_S}\T e^{2\pi i\rho^2/\tau}e^{- i\pi m^2}\tau^{m^2}\frac{\theta_1(Q+\rho|\tau)\theta_1(Q-\rho|\tau)}{\eta(\tau)^2}=-\frac{1}{\widehat{\Upsilon}_S}e^{- i\pi m^2}\tau^{m^2}e^{2\pi i\rho^2/\tau}Z_D.
\end{gather}
Now we consider this relation in terms of the conformal block $\cB$, with a shifted contour of integration in $\nu$ that gives it a large imaginary part:
\begin{align}
	\cB(a,m,\tau)  &=\int_0^1\dd\rho\int_{0+i\Lambda}^{4\pi+i\Lambda}\frac{\dd\nu}{4\pi}\,\eta(\tau)e^{-i \pi \tau/2} e^{-2\pi i (\rho +1/2)}Z_D(a,\nu,m,\tau,\rho) \nonumber\\
 &=-e^{i\pi m^2}\tau^{-m^2}\eta(\tau)e^{-i \pi \tau/2}\int_0^1\dd\rho\int_{0+i\Lambda}^{4\pi+i\Lambda}\frac{\dd\nu}{4\pi}\, e^{-2\pi i\rho^2/\tau}\widehat{\Upsilon}_S(\at,\nu) e^{-2\pi i (\rho +1/2)}\widetilde{Z}_D(a,\nu,m,\tau,\rho).
	\end{align}
As the next step, we use the expression
\begin{gather}
	e^{-2\pi i\rho^2/\tau}e^{4\pi i\left(k+\frac{n}{2}+\frac{1}{2} \right)\left(-\frac{\rho}{\tau}+\frac{1}{2} \right)}=e^{-\frac{2\pi i}{\tau}\left(\rho+k+\frac{n}{2}+\frac{1}{2} \right)^2}e^{\frac{2\pi i}{\tau}\left(k+\frac{n}{2}+\frac{1}{2} \right)^2}e^{2\pi ik+i\pi n+i\pi},
\end{gather}
in the definition of $Z_D$ \eqref{eq:ZDCB} to obtain
\begin{gather}
	\cB(a,m,\tau)  =-\frac{e^{i\pi m^2}\tau^{-m^2}\eta(\tau) e^{-i\pi \tau/2}}{\eta(\taut)}\int_0^1\dd\rho\int_{0+i\Lambda}^{4\pi+i\Lambda}\frac{d\nu}{4\pi}\,\widehat{\Upsilon}_S(\at,\nu)\sum_{n,k\in\mathbb{Z}}(-1)^n e^{\frac{i\nut n}{2}}e^{-\frac{2\pi i}{\tau}\left(k+\frac{n}{2}+\frac{1}{2} \right)^2}\nonumber \\
 \times \quad e^{-\frac{2\pi i}{\tau}\left(\rho+k+\frac{n}{2}+\frac{1}{2} \right)^2} e^{-2\pi i\rho} e^{\frac{2\pi i}{\tau}\left(k+\frac{n}{2}+\frac{1}{2} \right)^2}\cB\left(\at+\frac{n}{2},m,\taut \right) \nonumber\\
 =-\frac{e^{i\pi m^2}\tau^{-m^2}e^{-i\pi\tau/2}}{\sqrt{-i\tau}}\int_{0+i\Lambda}^{4\pi+i\Lambda} \frac{\dd\nu}{4\pi}\,\widehat{\Upsilon}_S(\at,\nu)\sum_{n}(-1)^ne^{\frac{i\nut n}{2}}\cB\left(\at+\frac{n}{2},m,\taut \right)\nonumber \\
 \times\left[\sum_k\int_0^1\dd\rho\, e^{-\frac{2\pi i}{\tau}(\rho+k+\frac{n}{2}+\frac{1}{2})^2} e^{-2\pi i \rho} \right].
\end{gather}
We now simplify the following term in the expression above
\begin{align}
    \sum_{k\in\mathbb{Z}}\int_0^1\dd\rho\, e^{-\frac{2\pi i}{\tau}(\rho+k+\frac{n}{2}+\frac{1}{2})^2} e^{-2\pi i \rho}&= \sum_{k\in\mathbb{Z}}\int_0^1\dd\rho\, e^{-\frac{2\pi i}{\tau}(\rho+\frac{n}{2}+\frac{1}{2}+\frac{\tau}{2})^2} e^{i\pi \tau/2} e^{2\pi i (k + \frac{n}{2} + \frac{1}{2})}\nonumber\\
    &=e^{i\pi \tau/2} e^{i\pi (n+1)}\int_{-\infty}^\infty \dd \rho\, e^{-\frac{2\pi i}{\tau}(\rho+\frac{n}{2}+\frac{1}{2}+\frac{\tau}{2})^2} \nonumber\\
    &=e^{i\pi \tau/2} e^{i\pi (n+1)}\sqrt{\frac{-i\tau}{2}},
\end{align}
along with the identity for $\eta(\tau)$
to find
\begin{gather}
   \cB(a,m,\tau)=\frac{e^{i\pi m^2}\tau^{-m^2}}{\sqrt{2}}\int_{0+i\Lambda}^{4\pi+i\Lambda}\frac{\dd\nu}{4\pi}\,\widehat{\Upsilon}_S(\at,\nu)\sum_n e^{\frac{in\nut}2} \cB\left(\at+\frac{n}{2},m,\taut \right) \nonumber \\
   =\frac{e^{i\pi m^2}\tau^{-m^2}}{\sqrt{2}}\int_{0+i\Lambda'}^{1+i\Lambda'}\frac{\dd\at}{4\pi}\,\frac{\partial\nu}{\partial\at}\,\widehat{\Upsilon}_S\left(\at,\nu(a,\at)\right)\sum_n e^{\frac{in\nut}2} \cB\left(\at+\frac{n}{2},m,\taut \right),
\end{gather}
where $\Lambda':=\Lambda/4\pi$. To see that the interval $[0+i\Lambda,4\pi+i\Lambda]$ in $\nu$ gets mapped to the interval $[0+i\Lambda',1+i\Lambda']$ in $\at$, we used \eqref{eq:aeqatmnu} and the second equation of \eqref{eq:nueqamnut} when $\at,\,\nu\rightarrow i\infty$: a shift of $\nu$ by $2\pi$ becomes a shift in $\at$ by $1/2$, while $\nut$ is periodic.  Using Lemma \ref{thm:UShift}, and switching summation with integration, we obtain
\begin{gather}
    \cB(a,m,\tau)=\frac{e^{i\pi m^2}\tau^{-m^2}}{\sqrt{2}}\int_{0+i\Lambda'}^{1+i\Lambda'}\frac{\dd\at}{4\pi}\,\frac{\partial\nu}{\partial\at}\,\widehat{\Upsilon}_S\left(\at,\nu(a,\at)\right)\sum_n e^{\frac{in\nut}2} \cB\left(\at+\frac{n}{2},m,\taut \right) \nonumber \\
    = \frac{e^{i\pi m^2}\tau^{-m^2}}{\sqrt{2}}\sum_{n\in\mathbb{Z}}\int_{0+i\Lambda'}^{1+i\Lambda'}\frac{\dd\at}{4\pi}\,\frac{\partial\nu}{\partial\at}\,e^{-\frac{in\nut}{2}}\widehat{\Upsilon}_S\left(\at+\frac{n}{2},\nu(a,\at)\right) e^{\frac{in\nut}2} \cB\left(\at+\frac{n}{2},m,\taut \right).
\end{gather}
Note that $\partial_{\at} \nu\left(\at+\frac{n}{2},a \right)=\partial_{\at} \nu\left(\at,a \right)$. Just as we did before, we can combine the summation and integration into a single integral\footnote{The shift is twice smaller than the integration domain, so each point is counted twice, which produces an extra factor of 2.}. Shifting the contour along the imaginary axis in order for the integral to be well-behaved, we obtain
\begin{equation}
    \cB(a,m,\tau)=\sqrt{2}(e^{-i\pi}\tau)^{-m^2}\int_{-\infty+i\Lambda'}^{\infty+i\Lambda'}\frac{\dd\at}{4\pi}\,\frac{\partial\nu}{\partial\at} \,\widehat{\Upsilon}_S(a,\at)\cB(\at,m,\taut),
\end{equation}
from which we can simply read off the modular kernel:
\begin{equation}
    S(a,\at)=\frac{\sqrt{2}}{4\pi}\frac{\partial\nu}{\partial\at} \,\widehat{\Upsilon}_S(a,\at),
\end{equation}
with the derivative of $\nu$ obtained simply from differentiating the two sides of \eqref{rel:MBs}.
\end{proof}
\section{Semiclassical modular kernel and complex Chern-Simons amplitudes}\label{sec:cinfty}

In this section, we will relate the $c=1$ modular kernel, that we computed in the previous section, to the $c\rightarrow\infty$ modular kernel, that we will now compute, and we will show that they represent two different generating functions of coordinate changes on the character variety. At the end of the section, we will give an interpretation of this result in terms of semiclassical $SL(2,\mathbb{C})$ Chern-Simons amplitudes. A completely analogous derivation can be done for the fusion kernel of four-point conformal blocks, and it would relate the $c\rightarrow\infty$ fusion kernel to the Painlev\'e VI connection constant computed in \cite{Iorgov2013,ILP2016}.

The Virasoro modular kernel for central charge $c\in\mathbb{C}\setminus (-\infty,1] $ was computed for the first time in \cite{Teschner2003} (see also \cite{Nemkov2015} for an alternative representation valid in the same region). Here we will consider the following expression, see \cite[eq. (3.88)]{Eberhardt2023}:
\begin{gather}\label{eq:PT}
	\mathcal{S}_{\sigma,\sigmat}(\mu)=
	\int\frac{\dd\zeta}{i}e^{4\pi i\zeta\sigma}\prod_{\e,\e'=\pm}S_b\left(\frac{Q}{4}+\frac{\mu}{2}+\e \sigmat +\e'\zeta\right).
\end{gather}
Here $b$ is the ``Liouville coupling constant'', related to the Virasoro central charge $c$ via
\begin{equation}
    c=1+6Q^2,\qquad Q:=b+b^{-1},
\end{equation}
and $S_b$ is the double-sine function,
\begin{equation}
    S_b:=\frac{\Gamma_b(z)}{\Gamma_b(Q-z)},
\end{equation}
defined through the double-gamma function $\Gamma_b$
\begin{equation}
    \log\Gamma_b(z):=\int_0^\infty\frac{\dd t}{t}\left(\frac{e^{\frac{t}{2}(Q-2z)}-1}{4\sinh\left(\frac{b t}{2}\right)\sinh\left(\frac{t}{2b}\right)}-\frac{1}{8}(Q-2z)^2e^{-t}-\frac{Q-2z}{2t} \right),\qquad \Re z>0.
\end{equation}
For more details on these functions we refer to the very thorough \cite[Appendix B]{Eberhardt2023}.
\begin{remark}
    To compare with the notation of loc. cit. set $\mu\equiv p_0,\,\sigma\equiv p_1,\,\sigmat\equiv p_2$. For simplicity, we have also stripped the kernel of an extra factor 
\begin{align}
\frac{\sqrt{2}}{S_b(2\sigmat)S_b(-2\sigmat)}\left(\frac{e^{\frac{i\pi}{2}\left(\frac{Q}{2}+\mu\right)\left(\frac{Q}{2}-\mu \right)}}{S_b\left(\frac{Q}{2}+\mu \right)}\frac{\Gamma_b(Q+2\sigma)\Gamma_b(Q-2\sigma)\Gamma_b\left(\frac{Q}{2}-\mu+2\sigmat \right)\Gamma_b\left(\frac{Q}{2}-\mu-2\sigmat \right)}{\Gamma_b\left(\frac{Q}{2}-\mu+2\sigma \right)\Gamma_b\left(\frac{Q}{2}-\mu-2\sigma \right)\Gamma_b(Q+2\sigmat)\Gamma_b(Q-2\sigmat)}\right). \label{extra-fac}
\end{align}
The semiclassical limit of the expression in round brackets above essentially reproduces the first line of \eqref{eq:ConnConst} with an extra $m$-dependent factor, while the first term has regular semiclassical limit.
\end{remark}
\begin{theorem}\label{thm:KernelGen}
The semiclassical $b\rightarrow0$ asymptotics of the $c>1$ Virasoro modular kernel $\mathcal{S}_{\sigma,\sigmat}$ from \eqref{eq:PT}, under further rescalings
\begin{equation}\label{eq:param_relabel}
\sigma=\frac{a}{b},\qquad\sigmat=\frac{\at}{b},\qquad \mu=\frac{m}{b}-\frac{Q}{2}, \qquad 4\pi\zeta=\frac{\nu}{b},
\end{equation}
is given by
    \begin{equation}
        \lim_{b\rightarrow 0}b^2 \log \mathcal{S}_{\sigma,\sigmat}(\mu)=-\mathcal{G},
    \end{equation}
    where $\mathcal{G}$ is the generating function defined in \eqref{def:G}, \eqref{def:G0}.
\end{theorem}
\begin{proof}
We will use the asymptotics \cite[eq. (B.37)]{Eberhardt2023}
\begin{equation}\label{eq:SbAs1}
	S_b\left(\frac{x}{b}\right)\approx \exp\left\{\frac{i}{4\pi b^2}\left[\Li_2(e^{2\pi ix})-\Li_2(e^{-2\pi ix}) \right]+O(b^0) \right\},
\end{equation}
where $\Li_2$ is the dilogarithm function, with integral representation
\begin{equation}
    \Li_2(z)=-\int_0^z\log(1-x)\frac{\dd x}{x},
\end{equation}
with choice of logarithmic cut $(1,\infty)$. 
One can use the inversion relation $\Li_2(z^{-1})=-\Li_2(z)-\frac{\pi^2}{6}-\frac{1}{2}\log^2(-z)$ of the dilogarithm, and the well-known relation between the dilogarithm and the Barnes G-function, \cite[above eq. (4.68)]{ILP2016}
\begin{equation}
	\Li_2(e^{2\pi ix})=-2\pi i\log\frac{G(1+x)}{G(1-x)}-2\pi ix\log\frac{\sin\pi x}{\pi}-\pi^2x(1-x)+\frac{\pi^2}{6}
\end{equation}
to rewrite \eqref{eq:SbAs1} as
\begin{equation}\label{eq:SbGhat}
    S_b\left(\frac{x}{b}\right)\approx \left[\widehat{G}(x)\left(\frac{\sin\pi x}{\pi} \right)^x \right]^{\frac{1}{b^2}}\times O(b^0).
\end{equation}
The nontrivial term is given by the integral:
\begin{equation}
\mathcal{I}:=\int\frac{\dd\zeta}{i}e^{4\pi i\sigma\zeta}\prod_{\e,\e'=\pm1}S_b\left(\frac{\mu+Q/2}{2}+\e \sigmab+\e'\zeta \right) 
=\int\frac{\dd\nu}{4\pi ib}e^{ia\nu/b^2}\prod_{\e,\e'=\pm1}S_b\left(\frac{1}{b}\left(\frac{m}{2}+\e\at+\e'\frac{\nu}{4\pi} \right) \right).
\end{equation}
We will compute the $b\rightarrow0$ through a saddle-point approach. To do this, it is convenient to use a slightly different form of the $b\rightarrow 0$ asymptotics of the double-sine function:
\begin{equation}\label{eq:SbClass}
S_b(x/b)\approx\exp\left\{\frac{i}{2\pi b^2}\left[\Li_2(e^{2\pi ix})-\frac{\pi^2}{6}-\pi^2x(x-1) \right] \right\},
\end{equation}
obtained from the previous one by using the previous inversion relation of the dilogarithm. Then the saddle-point equation takes the form
\begin{gather}
\frac{i}{b^2}\frac{\partial}{\partial\nu}\bigg\{a\nu+\frac{1}{2\pi}\bigg[\Li_2\left(e^{2\pi i\left(\frac{m}{2}+\at+\frac{\nu}{4\pi} \right)} \right)+\Li_2\left(e^{2\pi i\left(\frac{m}{2}-\at+\frac{\nu}{4\pi} \right)} \right)\nonumber \\
+\Li_2\left(e^{2\pi i\left(\frac{m}{2}+\at-\frac{\nu}{4\pi} \right)} \right)+\Li_2\left(e^{2\pi i\left(\frac{m}{2}-\at-\frac{\nu}{4\pi} \right)} \right)-\frac{\pi^2}{3}(3m^2-6m+12\at^2+2)-\frac{\nu^2}{4} \bigg] \bigg\}=0,
\end{gather}
and using $\partial_x\Li_2(e^{2\pi i x})=-2\pi i\log(1-e^{2\pi i x}) $, this is
\begin{equation}
a-\frac{\nu}{4\pi}+\frac{i}{4\pi}\log\left[\frac{\left(1-e^{2\pi i\left(\frac{m}{2}+\at -\frac{\nu}{4\pi} \right)} \right)\left(1-e^{2\pi i\left(\frac{m}{2}-\at -\frac{\nu}{4\pi} \right)} \right)}{\left(1-e^{2\pi i\left(\frac{m}{2}+\at +\frac{\nu}{4\pi} \right)} \right)\left(1-e^{2\pi i\left(\frac{m}{2}-\at +\frac{\nu}{4\pi} \right)} \right)} \right]=0,
\end{equation}
which is just the character variety relation \eqref{rel:a_as_atnu},
\begin{equation}
a=\frac{1}{4\pi i}\log\left[\frac{\sin\pi\left(\frac{m}{2}+\at-\frac{\nu}{4\pi} \right)\sin\pi\left(\frac{m}{2}-\at-\frac{\nu}{4\pi} \right)}{\sin\pi\left(\frac{m}{2}+\at+\frac{\nu}{4\pi} \right)\sin\pi\left(\frac{m}{2}-\at+\frac{\nu}{4\pi} \right)} \right],
\end{equation}
meaning that we can indeed identify the saddle-point values of $a,\at,\nu$ as the monodromy coordinates of the previous sections. 
We evaluate the integrand at the saddle point using \eqref{eq:SbGhat}, obtaining
\begin{gather}
\mathcal{I}\approx\frac{1}{4\pi ib}\bigg[e^{ia\nu}\widehat{G}\left(\frac{m}{2}+\at+\frac{\nu}{4\pi} \right)\widehat{G}\left(\frac{m}{2}-\at+\frac{\nu}{4\pi} \right)\widehat{G}\left(\frac{m}{2}+\at-\frac{\nu}{4\pi} \right)\widehat{G}\left(\frac{m}{2}-\at-\frac{\nu}{4\pi} \right) \nonumber \\
\times\frac{1}{\pi^{2m}}\prod_{\epsilon,\epsilon'=\pm}\sin\pi\left(\frac{m}{2}+\epsilon\at+\epsilon'\frac{\nu}{4\pi} \right)^{\frac{m}{2}+\epsilon\at+\epsilon'\frac{\nu}{4\pi}} \bigg]^{\frac{1}{b^2}} \\
=\frac{1}{4\pi ib}\left[\frac{e^{ia\nu}}{\pi^{2m}}\frac{\widehat{G}\left(\frac{m}{2}+\at+\frac{\nu}{4\pi} \right)\widehat{G}\left(\frac{m}{2}+\at-\frac{\nu}{4\pi} \right)}{\widehat{G}\left(\at-\frac{m}{2}+\frac{\nu}{4\pi} \right)\widehat{G}\left(\at-\frac{m}{2}-\frac{\nu}{4\pi} \right)}\prod_{\epsilon,\epsilon'=\pm}\sin\pi\left(\frac{m}{2}+\epsilon\at+\epsilon'\frac{\nu}{4\pi} \right)^{\frac{m}{2}+\epsilon\at+\epsilon'\frac{\nu}{4\pi}}\right]^{\frac{1}{b^2}} .
\end{gather}
Noting that the expression between square brackets is nothing but $e^{-\mathcal{G}}$, and taking the logarithm, we obtain the Theorem.
\end{proof}

The link between the connection constant and the $c\rightarrow\infty$ modular kernel is obtained through the following two functions \(\mathcal{G}\) and \(\mathcal{G}_0\).
\begin{proposition}\label{prop:GenFn}
Consider the following functions of monodromy data\footnote{The $m$ dependent factor comes through  the semiclassical limit of the term $\frac{e^{\frac{i\pi}{2}\left(\frac{Q}{2}+\mu\right)\left(\frac{Q}{2}-\mu \right)}}{S_b\left(\frac{Q}{2}+\mu \right)}$ in the expression \eqref{extra-fac} with a suitable multiplication by the term $\pi^{2m}$, which is simply a convenient normalization for what follows.}:
\begin{equation}
    e^{-\mathcal{G}_0}:=\frac{e^{-i \pi m^2/2}}{\widehat{G}(m) (\sin \pi m)^m } 
    \frac{\widehat{G}\left(\frac{m}{2}+\at+\frac{\nu}{4\pi} \right)\widehat{G}\left(\frac{m}{2}+\at-\frac{\nu}{4\pi} \right)}{\widehat{G}\left(\at-\frac{m}{2}+\frac{\nu}{4\pi} \right)\widehat{G}\left(\at-\frac{m}{2}-\frac{\nu}{4\pi} \right)}\prod_{\epsilon,\epsilon'=\pm}\sin\pi\left(\frac{m}{2}+\epsilon\at+\epsilon'\frac{\nu}{4\pi} \right)^{\frac{m}{2}+\epsilon\at+\epsilon'\frac{\nu}{4\pi}}, \label{def:G0}
\end{equation}
\begin{equation}
    \mathcal{G}:=\mathcal{G}_0-ia\nu.\label{def:G}
\end{equation}
These are generating functions on the character variety, namely
\begin{equation}\label{eq:geng0}
    \frac{\partial\mathcal{G}_0}{\partial\nu}=ia,\qquad \frac{\partial\mathcal{G}_0}{\partial \at}=i\nut,
\end{equation}
and
\begin{equation}\label{eq:geng}
    \frac{\partial\mathcal{G}}{\partial a}=-i\nu,\qquad\frac{\partial\mathcal{G}}{\partial \at}=i\nut.
\end{equation}
Furthermore, the generating function $\mathcal{G}_0$ is related to the connection constant $\widehat{\Upsilon}_S$ defined in \eqref{def:Upsilon1} in the following way: 
\begin{equation}
    e^{(\mathcal{G}_0-\widetilde{a}\partial_{\widetilde{a}}\mathcal{G}_0-m\partial_m\mathcal{G}_0)}= \widehat{\Upsilon}_S. \label{eq:G0Up}
\end{equation}
\end{proposition}
\begin{proof}
The proof is just a direct but very tedious computation of the derivatives, using the character variety relation in Proposition \ref{prop:SMonodromy}, that can be verified through the accompanying Mathematica file\footnote{We use the monodromy coordinate $\eta := \nu/4\pi$ in the mathematica file associated to this article for ease of computation.}. The sketch of the computation is as follows. 

We first verify the relations \eqref{eq:geng0} by direct substitution, and \eqref{eq:geng} follows by using \eqref{def:G} and \eqref{eq:geng0}. The relation to the connection constant follows two steps. In the first, we verify that the following relation holds
\begin{align}\label{ew:step1_prop7}
     e^{(\mathcal{G}_0-\widetilde{a}\partial_{\widetilde{a}}\mathcal{G}_0-m\partial_m\mathcal{G}_0)}= e^{i a \nu} \frac{  \left(\widehat{G}\left(\widetilde{a}+\nu/4\pi -\frac{m}{2}\right) \widehat{G}\left(\widetilde{a}-\nu/4\pi -\frac{m}{2}\right)\right)}{\left(\widehat{G}\left(\widetilde{a}+\nu/4\pi +\frac{m}{2}\right) \widehat{G}\left(\widetilde{a}-\nu/4\pi +\frac{m}{2}\right)\right)}(2 \pi )^m \widehat{G}(m) e^{-i\pi m^2/2},
\end{align}
and in the final step, using the character variety relations, we obtain the relation below 
\begin{align}
    \text{RHS of } \eqref{ew:step1_prop7} =  \widehat{\Upsilon}_S(a, \widetilde{\nu}),
\end{align}
which proves the relation \eqref{eq:G0Up}.
\end{proof}

In this way, we have obtained exact relations between $c=1$ and $c=\infty$ Virasoro modular kernels, through their relation with the modular transformations of isomonodromic tau functions. Let us conclude with a few remarks.

\begin{remark}
    Recall the general definition of the modular kernel \eqref{eq:Skerneldef}, that we recall here for convenience:
    \begin{equation}
	\cB(\sigma,\mu,\tau)=(e^{-i\pi}\tau)^{-\Delta(m)}\int_{-\infty+i\Lambda}^{\infty+i\Lambda}\dd \widetilde{\sigma}\,\mathcal{S}_{\sigma,\sigmat}(\mu)\widetilde{\cB}(\sigmat,m,\taut).
\end{equation}
    Taking the semi-classical limit, one obtains from a saddle-point computation the following relations
    \begin{equation}
        \frac{\partial}{\partial\at}\left(b^2\lim_{b\rightarrow0}\log\cBt(\at)\right)=\frac{\partial}{\partial\at}\lim_{b\rightarrow 0}b^2\log \mathcal{S}_{\sigma,\sigmat}(\mu),\qquad         \frac{\partial}{\partial a}\left(b^2\lim_{b\rightarrow0}\log(\mathcal{S}_{\sigma,\sigmat}(\mu)\cB(a))\right)=0,
    \end{equation}
    
    where $\sigma=a/b$ and $\sigmat=\at/b$, and $\mu=\frac{m}{b}-\frac{Q}{2}$. On the other hand, we already showed in Theorem \ref{thm:KernelGen} that
\begin{equation}
    \lim_{b\rightarrow 0}b^2\log \mathcal{S}_{\sigma,\sigmat}(\mu)=-\mathcal{G}(a,\at),\quad \frac{\partial \mathcal{G}}{\partial a}=-i\nu,\quad \frac{\partial \mathcal{G}}{\partial\at}=i\nut.
\end{equation}
If we define $F_{NS}(a):=\lim_{b\rightarrow 0} b^2\mathcal{B}(\sigma,\mu,\tau) $ and $\widetilde{F}_{NS}(\at):=\lim_{b\rightarrow 0} b^2\mathcal{\widetilde{B}}(\sigmat,\mu,\taut) $, we find that
\begin{equation}
    \frac{\partial}{\partial \at} \widetilde{F}_{NS}(\at)=i\nut,\qquad \frac{\partial}{\partial a}F_{NS}(a)=i\nu,
\end{equation}
consistent with the semiclassical limit of so-called ``Blowup relations'' for conformal blocks\footnote{The name ``blowup relation'' stems from their original derivation in the context of supersymmetric partition functions on manifolds with a blown-up point \cite{Nakajima2003}.} from \cite{Bershtein2021}. In fact, a rigorous proof of this statement also follows from the probabilistic approach to conformal blocks on the torus \cite{ghosal2020probabilistic,desiraju2024proof} and will be the subject of an upcoming paper by one of the authors.
\end{remark}

\begin{remark}
The explicit form of \eqref{eq:G0Up} can also be obtained from the blowup relations from \cite[7.23]{Bershtein2021} using the idea from \cite[sec. 7.4]{OlegMovie,Gavrylenko:2025nuo}.
Namely, we can study the connection constant for the tau functions specialised to the point where \(Q\) vanishes.
This point is described by the equation \(\frac{\partial}{\partial a}F_{NS}(a)=i\nu\).
The value of the tau function itself is described by
\begin{multline}
\widehat{\mathcal{T}}(a,m,\tau)=\exp\left(F^{NS}(a,m,\tau)-m\partial_m F^{NS}(a,m,\tau)-a\partial_aF^{NS}(a,m,\tau)\right)=\\=
\exp\left(F^{NS}(a,m,\tau)-m\partial_m F^{NS}(a,m,\tau)-ia\nu\right).
\end{multline}
If we now substitute the transformation property of classical conformal blocks
\begin{equation}
F^{NS}(a,m,\tau)=m^2\log(e^{-\pi i}\tau)+F^{NS}(\widetilde{a},m,\widetilde{\tau})-\mathcal{G}(a,\widetilde{a},m)
\end{equation}
into this formula, we will get the following sequence of equalities
\begin{equation}
\begin{gathered}
\begin{split}
&\widehat{\mathcal{T}}(a,m,\tau)=\exp \left(F^{NS}(\widetilde{a},m,\widetilde{\tau})-m\partial_m F^{NS}(\widetilde{a},m,\widetilde{\tau}) -\mathcal{G}(a,\widetilde{a},m)+m\partial_m \mathcal{G} -i a\nu -m^2\log(e^{-\pi i}\tau)\right)
\\&=
(e^{-\pi i}\tau)^{-m^2}\exp \left(F^{NS}(\widetilde{a},m,\widetilde{\tau})-m\partial_m F^{NS}(\widetilde{a},m,\widetilde{\tau}) -\mathcal{G}_0(\nu,\widetilde{a},m)+m\partial_m \mathcal{G}(a,\widetilde{a},m) \right)
\\&=
(e^{-\pi i}\tau)^{-m^2}\exp \left(F^{NS}(\widetilde{a},m,\widetilde{\tau})-m\partial_m F^{NS}(\widetilde{a},m,\widetilde{\tau}) -i \widetilde{a}\widetilde{\nu}\right)
\exp\left( -\mathcal{G}_0(\nu,\widetilde{a},m)+m\partial_m \mathcal{G}_0(\nu,\widetilde{a},m)+i\widetilde{a}\widetilde{\nu} \right)
\\&=
e^{\pi im^2}\tau^{-m^2}\widehat{\mathcal{T}}(\widetilde{a},m,\widetilde{\tau}) \exp \left( -\mathcal{G}_0(\nu,\widetilde{a},m)+m\partial_m \mathcal{G}_0(\nu,\widetilde{a},m)+\widetilde{a}\partial_{\widetilde{a}} \mathcal{G}_0(\nu,\widetilde{a},m) \right),
\end{split}
\end{gathered}
\end{equation}
which precisely reproduces the formula \eqref{eq:SDualTau}, after setting $Q=0$ and using \eqref{eq:G0Up}. The tau function defined in this way has a different normalisation than the one we used throughout the paper, which leads to a connection constant $\widehat{\Upsilon}_S$ given by \eqref{eq:Uhat}, instead of the full expression $\Upsilon_S$ of \eqref{eq:ConnConst}. 
\end{remark}

\begin{remark}
    It has been known since the early 1990s \cite{Reshetikhin1992,Harnad1994} that isomonodromic deformations can be understood as arising from the semi-classical limit of the Knzhnik-Zamolodchikov-Bernard (KZB) equations \cite{Knizhnik1984,Bernard1987,Bernard1988}. From the perspective of conformal field theory, these equations originate from an underlying affine current algebra, which in our case would be $\widehat{\mathfrak{sl}_2}$. In this framework, the isomonodromic tau function is the semiclassical wave function of the KZB equations. 
    
    On the other hand, KZB equations are also expected to arise as the equations governing the complex Chern-Simons partition function on the space $C_{g,n}\times S^1$, where $C_{g,n}$ is a Riemann surface of genus $g$ with $n$ marked points \cite{Witten1989}. More recently, the modular kernel has been also studied from the point of view of complex Chern-Simons theory, and it is given by the partition function on a 3-manifold formed by a cylinder connecting two tori with dual pants decompositions \cite{Dimofte2013}. This 3-manifold is called the mapping cylinder $M_S$, and it is the complement of the Hopf link in the 3-sphere $S^3$ (for a pictorial representation of this space, see \cite{Dimofte2015}, Figure 37). Quite naturally from this perspective, if tau functions are semiclassical limits of wave-functions, then connection constants are semiclassical limits of transition amplitudes.
\end{remark}

\appendix
\section{Character variety}\label{sec:charvar}
In this section, we explain the character variety of the one point torus. Let us start by listing the monodromies on the one-point torus:
\begin{align}
M_A=e^{2\pi ia\sigma_3}, && M_0 \sim e^{2\pi i m \sigma_3},  && M_B=\frac{1}{\sin2\pi a}\left( \begin{array}{cc}
\sin(\pi(2a-m))e^{-i\eta/2} & \sin(\pi m)e^{i\eta/2} \\
\sin(\pi m)e^{-i\eta/2} & \sin(\pi(2a+m))e^{i\eta/2} 
\end{array} \right).
\end{align}
The space of monodromies is then defined as follows
\begin{equation}
\mathcal{M}_{1,1}:=\left\{M_A,M_B\in SL(2):\, \left(M_AM_BM_A^{-1}M_B^{-1} \right)=M_0 \right\}/\sim,
\end{equation}
where $/\sim$ means up to conjugation. 

In order to obtain the cubic equation of the character variety, we define the following coordinates:
\begin{align}
A:=\tr M_A, && B:=\tr M_B, && C:=\tr M_AM_B,
\end{align}
and explicitly,
\begin{align}
A &= e^{2\pi i a} + e^{-2\pi i a} \nonumber \\
B &=\frac{\sin(\pi(2a-m))}{\sin2\pi a}e^{-i\eta/2}+\frac{\sin(\pi(2a+m))}{\sin2\pi a}e^{i\eta/2}, \nonumber, \\
C &=\frac{\sin(\pi(2a-m))}{\sin2\pi a}e^{i(2\pi a-\eta/2)}+\frac{\sin(\pi(2a+m))}{\sin2\pi a}e^{-i(2\pi a-\eta/2)}.
\end{align}
The trace of the constraint on the monodromies is then given by
\begin{equation}
\begin{split}
\tr M_0=\tr M_AM_BM_A^{-1}M_B^{-1}&=\tr M_AM_BM_A^{-1}\tr M_B-\tr M_AM_BM_A^{-1}M_B\\
&=(\tr M_B)^2-\tr M_AM_BM_A^{-1}M_B \\
& =(\tr M_B)^2-\tr(M_AM_B)\tr(M_B^{-1}M_A)+\tr(M_AM_BM_B^{-1}M_A)\\
&=(\tr M_B)^2-\tr(M_AM_B)\tr(M_AM_B^{-1})+\tr(M_A^2) \\
&=(\tr M_B)^2-\tr(M_AM_B)\left[\tr M_A\tr M_B-\tr(M_AM_B) \right]+\tr(M_A)^2-2 \\
&=B^2-ABC+C^2+A^2-2,
\end{split}
\end{equation}
where we used the following identities for $SL(2)$ traces
\begin{align}\label{eq:SL2id}
\tr(xy)+\tr(xy^{-1})=\tr x\tr y, && \tr(x^2)=(\tr x)^2-2.
\end{align}
Therefore, the character variety is described by the cubic curve\footnote{To match with the conventions of \cite{NRS2011}, set $e^{2\pi i m}+ e^{-2\pi i m} \equiv \mathsf{m}$.}
\begin{equation}
\mathcal{W}_{1,1}=A^2+B^2+C^2-ABC-(e^{2\pi i m}+ e^{-2\pi i m})-2=0.
\end{equation}
Furthermore, the functions $A,B,C$ satisfy the Goldman bracket:
\begin{equation}
\{A,B\}=C-\frac{1}{2}AB,
\end{equation}
and $a, \eta$ are the Darboux coordinates
\begin{equation}
\{a,2\pi\eta \}=1.
\end{equation}
Under modular transformation, one obtains the dual monodromy coordinates $\widetilde{a}, \widetilde{\eta}$, and (upto factors of $2\pi$)
\begin{align}
    \eta&=\widetilde{a}-2i\log \frac{\sin( a- m/2-\widetilde{\eta})}{\sin(a+ m/2-\widetilde{\eta})}.
\end{align}

\section{Modular transformation for zero mass}\label{sec:zerom}
Here, we give a list of formulas to explain what happens in the case of free particle (\(m=0\)) and to provide an extra consistency check of our computations.

Solution of the linear system \eqref{linear_systemCM} is given by
\begin{equation}
Y(z,\tau)= \begin{pmatrix} e^{2\pi i \frac{dQ}{d\tau} z} & 0\\0 & e^{-2\pi i \frac{dQ}{d\tau} z} \end{pmatrix}.
\end{equation}
Its monodromies are
\begin{equation}
M_A= \begin{pmatrix} e^{2\pi i \frac{dQ}{d\tau}} & 0\\0 & e^{-2\pi i \frac{dQ}{d\tau}} \end{pmatrix},
\qquad
M_B= \begin{pmatrix} e^{2\pi i (\tau\frac{dQ}{d\tau}-Q)} & 0\\0 & e^{-2\pi i (\tau\frac{dQ}{d\tau}-Q)} \end{pmatrix}.
\end{equation}
Comparing these monodromies with \eqref{eq:MAM0MB} we identify parameters of solution:
\begin{equation}
Q=a \tau+\frac{\nu}{4\pi}.
\end{equation}
Using the definition \eqref{def:Tau} we find the tau function
\begin{equation}
\mathcal{T}=e^{2\pi i a^2 \tau}.\label{eq:this3}
\end{equation}
Dual partition function \eqref{eq:ZDCB} in this case becomes
\begin{equation}
\label{eq:ZDm0}
Z_D(a,\nu,0,\tau,\rho)=\frac{e^{2\pi i a^2 \tau}}{\eta(\tau)^2}\theta_1\left(a \tau+\frac{\nu}{4\pi}+\rho\middle|\tau\right)\theta_1\left(a \tau+\frac{\nu}{4\pi}-\rho\middle|\tau\right).
\end{equation}

Modular transformation of solution described by \eqref{eq:modtraPQ} reads
\begin{equation}
\widetilde{Q}=-Q/\tau=-a-\frac{\nu}{4 \pi \tau}=\frac{\nu}{4\pi}\widetilde{\tau}-a.
\end{equation}
This defines a mapping between parameters
\begin{equation}
\widetilde{a}=\frac{\nu}{4\pi},\qquad \frac{\widetilde{\nu}}{4\pi}=-a,
\end{equation}
which is consistent with the equations \eqref{rel:MAs}, \eqref{rel:MBs}, \eqref{rel:MABs}.

Modular transformation of the tau function \eqref{eq:SDualTau} reads
\begin{multline}
\widetilde{\mathcal{T}}^A(\widetilde{\tau};\widetilde{a},\widetilde{\nu},0)=\frac1{\Upsilon_S(a,\nu,0)}\mathcal{T}^A(\tau;a,n,0)e^{-2\pi i Q(\tau;a,\nu)^2/\tau}=
\frac1{\Upsilon_S(a,\nu,0)}e^{2\pi i a^2\tau}e^{-2\pi i \left(a\tau+\frac{\nu}{4\pi}\right)^2/\tau}=\\=\frac{e^{-i\nu a}}{\Upsilon_S(a,\nu,0)}e^{2\pi i \left( \frac{\nu}{4\pi} \right)^2\widetilde{\tau}} = \frac{e^{-i\nu a}}{\Upsilon_S(a,\nu,0)} \mathcal{T}^A(\widetilde{\tau};\widetilde{a},\widetilde{\nu},0)
\end{multline}
Looking at this formula, we conclude that
\begin{equation}
\Upsilon_S(a,\nu,0)=\widehat{\Upsilon}_S(a,\nu,0)=e^{-i\nu a}=e^{i\widetilde{\nu} \widetilde{a}},
\end{equation}
which is consistent with \eqref{def:Upsilon}.

The modular kernel \eqref{eq:c1modular} in this case becomes\footnote{To compare with \cite[eq. (2.55)]{Eberhardt2023}, note that our modular kerel involves an integral from $-\infty$ to $+\infty$, while in loc. cit. the integral is from $0$ to $\infty$.}
\begin{equation}
S(a,\tilde{a})=\sqrt{2}e^{-4\pi i a\widetilde{a}}.
\end{equation}
One can also recover the Virasoro conformal block in the following way. Combining \eqref{eq:this1}, \eqref{eq:this2}, and \eqref{eq:this3} we find
\begin{align}
    \mathcal{B}(a,0,\tau) = \int_0^1 \dd\rho \int_0^{4\pi} \frac{\dd\nu}{4\pi} \frac{e^{2\pi i a^2 \tau}}{\eta(\tau)} e^{-i \pi \tau/2} e^{-2\pi i (\rho+1/2)} \theta_1(a\tau + \frac{\nu}{4\pi}+ \rho)\theta_1(a\tau + \frac{\nu}{4\pi}- \rho).
\end{align}
With the series expression for the theta function
\begin{align}
    \theta_1(z) = - i \sum_{n\in \mathbb{Z}} (-1)^n q^{(n+1/2)^2/2} e^{2\pi i z (n+1/2)},
\end{align}
the above expression is
\begin{align}
   & \mathcal{B}(a,0,\tau) 
   =\int_0^1 \dd\rho \int_0^{4\pi} \frac{\dd\nu}{4\pi} \frac{e^{2\pi i a^2 \tau}}{\eta(\tau)}  \nonumber\\
   &\times
\sum_{m,n\in \mathbb{Z}}(-1)^{m+n+1} q^{(m^2 + n^2)/2+ (m+n)/2}\textrm{exp}\left(  2\pi i( (a\tau + \frac{\nu}{4\pi})(m+n+1) + \rho(m-n+1)) \right) \nonumber\\
   & =  \frac{e^{2\pi i a^2 \tau}}{\eta(\tau)}.
\end{align}
In the last line, we use the fact that the only terms that survive the integration are constant in $\nu, \rho$. 
Therefore,
\begin{equation}
\mathcal{B}(a,0,\tau)=\frac{e^{2\pi i a^2 \tau}}{\eta(\tau)}.
\end{equation}
The above expression can also be derived by using the AGT relation \cite[eq. (D.10)]{Bonelli2019}.
On the CFT side, it is simply the character of the Virasoro Verma module with the highest weight \(a^2\).

Modular transformation of \(m=0\) conformal block is defined by \eqref{eq:Skerneldef} and equals explicitly
\begin{equation}
\int_{-\infty}^{\infty}d\widetilde{a}\, \mathcal{B}(\widetilde{a},0,\tau)S(a,\widetilde{a}) = \int_{-\infty}^{\infty} d\widetilde{a}\, \frac{e^{-2\pi i \widetilde{a}^2/\tau}}{\eta(-1/\tau)} \sqrt{2} e^{-4\pi i a\widetilde{a}} = \frac{e^{2\pi i a^2\tau}}{\eta(-1/\tau)/\sqrt{-i\tau}} = \mathcal{B}(a,0,\tau),
\end{equation}
which reproduces the desired expression.

Another object that is worth to be computed explicitly in the \(m=0\) limit is the Fredholm determinant \cite[eq. (1.17)]{DMDG2020}:
\begin{multline}
K_{1,1}(z,w;\tau)=
\begin{pmatrix}\frac{-e^{-2\pi i\rho+2\pi i Q \sigma_3}}{e^{2\pi i(z-w+\tau)}-1} &0 \\ 0& \frac{e^{2\pi i\rho-2\pi i Q\sigma_3}}{e^{2\pi i(z-w-\tau)}-1}\end{pmatrix}
=\\=
\begin{pmatrix} \sum_{n=0}^{\infty}e^{-2\pi i\rho+2\pi iQ\sigma_3+2\pi n\tau}e^{2\pi i n(z-w)} & 0 \\ 0 & \sum_{n=1}^{\infty}e^{2\pi i\rho-2\pi iQ\sigma_3+2\pi in\tau}e^{2\pi in(w-z)}\end{pmatrix}.
\end{multline}
We see that this matrix is already diagonal in the Fourier basis, so its determinant equals
\begin{multline}
\det(1-K_{1,1})=\prod_{\epsilon=\pm1}\prod_{n=0}^{\infty}(1-e^{2\pi in\tau}e^{2\pi i(\epsilon Q-\rho)})\prod_{n=1}^{\infty}(1-e^{2\pi in\tau}e^{-2\pi i(\epsilon Q-\rho)})=\\=
e^{-2\pi i\rho}\prod_{\epsilon=\pm1}(e^{\pi i(\rho-\epsilon Q)}-e^{-\pi i(\rho-\epsilon Q)})
\prod_{n=1}^{\infty}(1-e^{2\pi in\tau}e^{2\pi i(\epsilon Q-\rho)})(1-e^{2\pi in\tau}e^{-2\pi i(\epsilon Q-\rho)})
=\\=
\frac{q^{1/12}}{\eta(\tau)^2} e^{-2\pi i\rho} q^{-1/4}\prod_{\epsilon=\pm1}\theta_1(\rho-\epsilon Q|\tau)=\frac{q^{-1/6}}{\eta(\tau)^2}e^{-2\pi i\rho}\theta_1(\rho+Q|\tau)\theta_1(\rho-Q|\tau).
\end{multline}
We see that it is also consistent with the dual partition function \eqref{eq:ZDm0}.

\bibliographystyle{JHEP}
\bibliography{Biblio}

\end{document}